\documentclass[twoside,leqno]{article}

\usepackage[letterpaper]{geometry}

\usepackage{siamproceedings}

\usepackage[T1]{fontenc}
\usepackage{amsfonts}
\usepackage{graphicx}
\usepackage{epstopdf}
\usepackage{enumitem}
\usepackage{color}
\usepackage{amssymb}
\usepackage{amsmath}
\usepackage{listings}
\usepackage[noend]{algpseudocode}
\usepackage{algorithm}
\ifpdf
  \DeclareGraphicsExtensions{.eps,.pdf,.png,.jpg}
\else
  \DeclareGraphicsExtensions{.eps}
\fi

\definecolor{codegreen}{rgb}{0,0.6,0}
\definecolor{codegray}{rgb}{0.5,0.5,0.5}
\definecolor{codepurple}{rgb}{0.58,0,0.82}
\definecolor{backcolour}{rgb}{0.95,0.95,0.92}

\lstdefinestyle{mystyle}{
  backgroundcolor=\color{backcolour},   commentstyle=\color{codegreen},
  keywordstyle=\color{magenta},
  numberstyle=\tiny\color{codegray},
  stringstyle=\color{codepurple},
  basicstyle=\ttfamily\footnotesize,
  breakatwhitespace=false,         
  breaklines=true,                 
  captionpos=b,                    
  keepspaces=true,                 
  numbers=left,                    
  numbersep=5pt,                  
  showspaces=false,                
  showstringspaces=false,
  showtabs=false,                  
  tabsize=2
}

\newsiamremark{remark}{Remark}
\newsiamremark{hypothesis}{Hypothesis}
\crefname{hypothesis}{Hypothesis}{Hypotheses}
\newsiamthm{claim}{Claim}

\usepackage{amsopn}

\begin{document}

\newcommand\relatedversion{}

\title{\Large Generating pivot Gray codes for spanning trees of complete graphs in constant amortized time\relatedversion}
    \author{Bowie Liu\thanks{Faculty of Applied Sciences, Macao Polytechnic University (\email{bowen.liu@mpu.edu.mo}, \email{cwong@mpu.edu.mo}, \email{ctlam@mpu.edu.mo}).}
    \and Dennis~Wong\footnotemark[1]    
    \and Chan-Tong Lam\footnotemark[1] 
    \and Sio-Kei Im\thanks{Macao Polytechnic University
  (\email{marcusim@mpu.edu.mo}).}   
}

\date{}

\maketitle


\fancyfoot[R]{\scriptsize{Copyright \textcopyright\ 2026 by SIAM\\
Unauthorized reproduction of this article is prohibited}}





\begin{abstract} 
We present the first known pivot Gray code for spanning trees of complete graphs, listing all spanning trees such that consecutive trees differ by pivoting a single edge around a vertex. 
This pivot Gray code thus addresses an open problem posed by Knuth in \emph{The Art of Computer Programming, Volume 4} (Exercise 101, Section 7.2.1.6, [Knuth, 2011]), rated at a difficulty level of 46 out of 50, and imposes stricter conditions than existing revolving-door or edge-exchange Gray codes for spanning trees of complete graphs. 
Our recursive algorithm generates each spanning tree in constant amortized time using $O(n^2)$ space. 
In addition, we provide a novel proof of Cayley's formula, $n^{n-2}$, for the number of spanning trees in a complete graph, derived from our recursive approach. 
We extend the algorithm to generate edge-exchange Gray codes for general graphs with $n$ vertices, achieving $O(n^2)$ time per tree using $O(n^2)$ space. 
For specific graph classes, the algorithm can be optimized to generate edge-exchange Gray codes for spanning trees in constant amortized time per tree for complete bipartite graphs, $O(n)$-amortized time per tree for fan graphs, and $O(n)$-amortized time per tree for wheel graphs, all using $O(n^2)$ space.
\end{abstract}

\section{Introduction.}\label{sec:intro}
A \emph{complete graph} on $n$ vertices, denoted by $K_n$, is a simple labeled graph in which every pair of distinct vertices is connected by exactly one edge. 
In other words, each vertex in a complete graph is adjacent to every other vertex. 
For example, the complete graphs on one to five vertices, that is $K_1$ to $K_5$, are shown in Figure~\ref{fig:k1k5}:

\begin{figure}[H]
    \centering
    \includegraphics[width=0.9\linewidth]{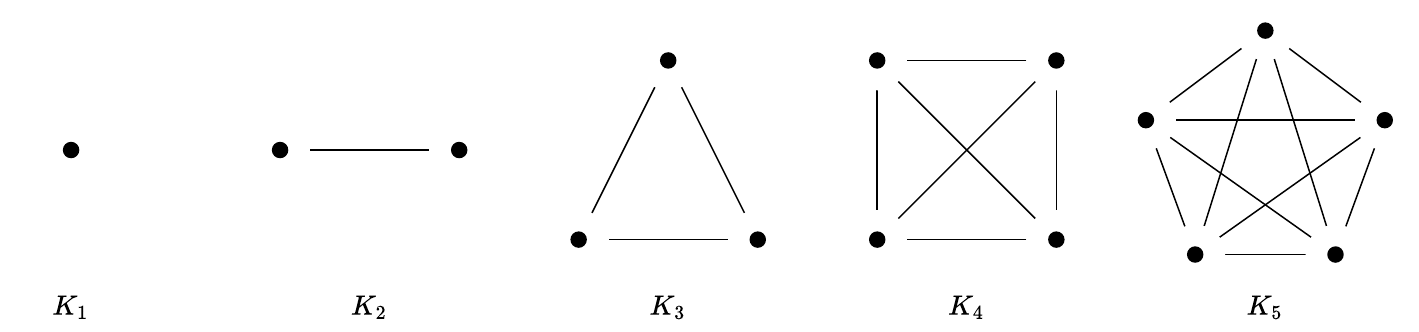}
        \caption{Complete graphs from $K_1$ to $K_5$.}
    \label{fig:k1k5}
\end{figure}

A \emph{spanning tree} of a simple graph $G$ is a subgraph that includes all vertices of $G$ and forms a tree. 
In other words, it is a subgraph of $G$ that connects all vertices of $G$ without cycles. 
As an example, Figure~\ref{fig:span16} illustrates the 16 spanning trees of $K_4$. 

The number of spanning trees in a complete graph $K_n$ with $ n $ vertices is given by Cayley's formula~\cite{Cayley1889}. 
This formula, named after the mathematician Arthur Cayley, gives the number of labeled spanning trees in $K_n$ as:
$$n^{n-2}.$$
This enumeration sequence corresponds to sequence A000272 in the Online Encyclopedia of Integer Sequences~\cite{OEIS2025A000272}. 
For example, the number of spanning trees for complete graphs $K_1$ to $K_8$ is 1, 1, 3, 16, 125, 1296, 16807, and 262144, respectively.
Cayley provided an informal proof of the formula for the number of spanning trees in a complete graph $K_n$~\cite{Cayley1889}. 
Later, Pr\"{u}fer~\cite{Prufer1918} and Kirchhoff~\cite{Kirchhoff1847} offered simpler proofs by establishing a one-to-one correspondence between spanning trees and Pr\"{u}fer sequences, and by applying the matrix-tree theorem, respectively.
M\"{u}tze~\cite{mutze2023combinatorial}, however,  commented that in \emph{The Art of Computer Programming, Volume 4}~\cite{knuth2011} (Ex. 101 in Sec. 7.2.1.6), Knuth possibly sought a more fine-grained explanation of Cayley's formula for the number of spanning trees in complete graphs.

The study of spanning trees of complete graphs and other graphs has a long history and has traditionally attracted considerable interest from mathematicians and graph theorists~\cite{Berger1967, Cayley1889, Chakraborty2019, Gabow1978, Kirchhoff1847,  knuth2011, Mayeda1965, Minty1965, mutze2023combinatorial, Naddef1981, Prufer1918, ruskey1996combinatorial, Sav97, Shank1968, Shioura1997, Winter1985}.  
The task of efficiently enumerating all spanning trees of graphs is fundamental in computer science and finds applications across diverse scientific fields including epidemiology, networking, chemistry, biology, astronomy, archaeology, image processing, and social media. Notably, one of the earliest known works on enumerating all spanning trees of a graph was conducted by the German physicist Wilhelm Feussner in 1902, motivated by applications in electrical network analysis \cite{feussner1902}. 
For more applications on spanning trees of graphs, see \cite{Chakraborty2019, knuth2011}.
\begin{figure}[t]
    \centering
    \includegraphics[width=1\linewidth]{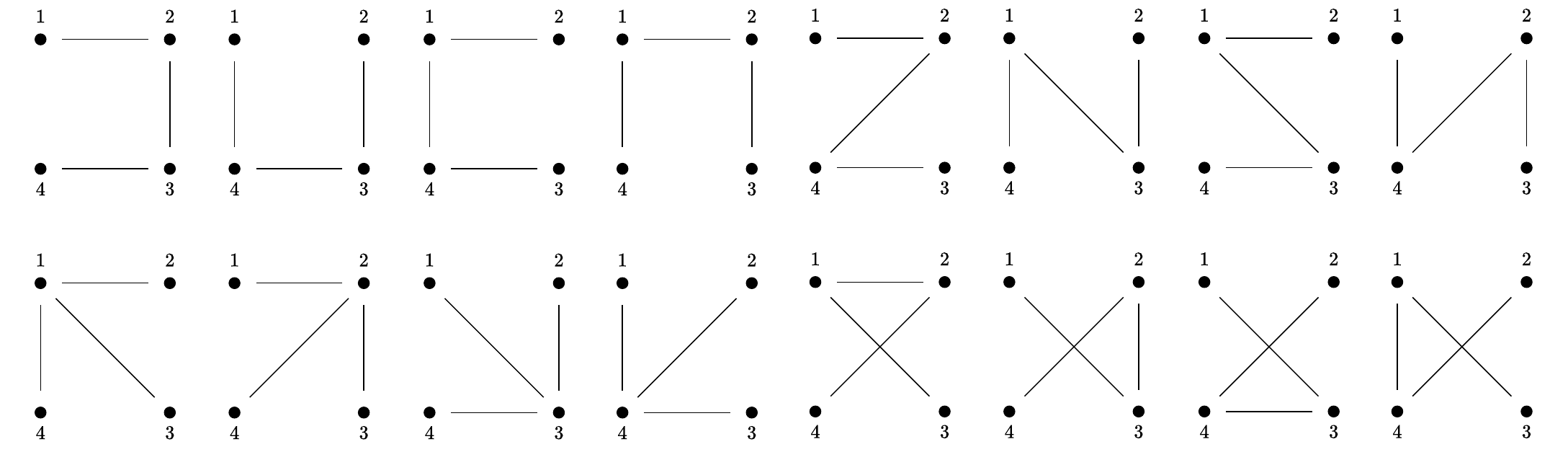}
    \caption{The 16 spanning trees of $K_4$.}
    \label{fig:span16}
\end{figure}

One of the most important aspects of combinatorial generation is to list the instances of a combinatorial object so that consecutive instances differ by a specified \emph{closeness condition} involving a constant amount of change. 
Lists of this type are called \emph{Gray codes}. 
This terminology is due to the eponymous \emph{binary reflected Gray code} (BRGC) by Frank Gray, which orders the $2^n$ binary strings of length $n$ so that consecutive strings differ by one bit. 
For example, when $n = 4$ the order is
$$0000, 1000, 1100, 0100, 0110, 1110, 1010, 0010, 0011, 1011, 1111, 0111, 0101, 1101, 1001, 0001.$$ 
Variations that reverse the entire order or the individual strings are also commonly used
in practice and in the literature. 
For more information about binary reflected Gray code and its variations, see~\cite{10.1007/978-3-030-85088-3_15,flipswap,Vajnovszki2007GrayCO}. 
We note that the order above is \emph{cyclic} because the last and first strings also differ by the closeness condition, and this property holds for all $n$.
The BRGC listing is a \emph{1-Gray code} in which consecutive strings differ by one symbol change. 
In this paper, we utilize a generalized version of the binary reflected Gray code to construct Gray codes for spanning trees of complete graphs and general graphs.

An interesting problem related to graphs is thus to discover a Gray code for spanning trees of graphs.
But before addressing this problem, we first define the nature of changes allowed in Gray codes for spanning trees of graphs. 
An \emph{edge-exchange Gray code} for the spanning trees of a graph is a list of all spanning trees of the graph such that each consecutive spanning trees differ by the removal of one edge and the addition of another. 
For complete graphs, this problem was posed as an open problem by Knuth in \emph{The Art of Computer Programming, Volume 4} \cite{knuth2011} (Ex. 101 in Sec. 7.2.1.6), with a difficulty rating of 46 out of 50:
\begin{itemize}
\item Is there a simple revolving-door way to list all $n^{n-2}$ spanning trees of the complete graph $K_n$? (The order produced by Algorithm S is quite complicated.) 
\end{itemize}
The \emph{revolving-door} property, as described by Knuth, refers to an edge-exchange Gray code in which consecutive spanning trees differ by the substitution of one edge for another \cite{knuth2011}. 
Cummins first demonstrated that the spanning trees of any connected graph can be listed cyclically via edge exchanges, thereby proving the existence of edge-exchange Gray codes for spanning trees of general graphs \cite{cummins1966}. 
Later, Kamae \cite{kamae19777}, Kishi and Kajitani \cite{kishi1967}, and Holzmann and Harary \cite{holzmann19721} provided more detailed proofs of the existence of edge-exchange Gray codes for spanning trees of graphs, with extensions to matroids. 
However, these results are existential and do not provide algorithms for generating such Gray codes. 
There are also many algorithms for generating spanning trees of graphs~\cite{Aichholzer2007, Berger1967, Chakraborty2019, Char1968, Gabow1978, Hakimi1961, Kapoor1995, Kapoor2000, Katoh2009a, Katoh2009b, Matsui1997, Mayeda1965, Minty1965, Shioura1995, Winter1985}, however, the spanning trees generated by these algorithms do not result in an edge-exchange Gray code. 
The problem of generating an edge-exchange Gray code for spanning trees of an arbitrary graph was resolved by Smith~\cite{smith1997}, and more recently by Merino, M\"{u}tze, and Williams~\cite{Merino2022}, and Merino and M\"{u}tze~\cite{doi:10.1137/23M1612019}.
Notably, Smith’s algorithm corresponds to the Algorithm S referenced in Knuth’s open problem.
Their algorithm produces the edge-exchange Gray code for $K_5$ as shown in Figure~\ref{fig:smith}:
\begin{figure}[h]
    \centering
    \includegraphics[width=1\linewidth, trim={25 0 75 0}, clip]{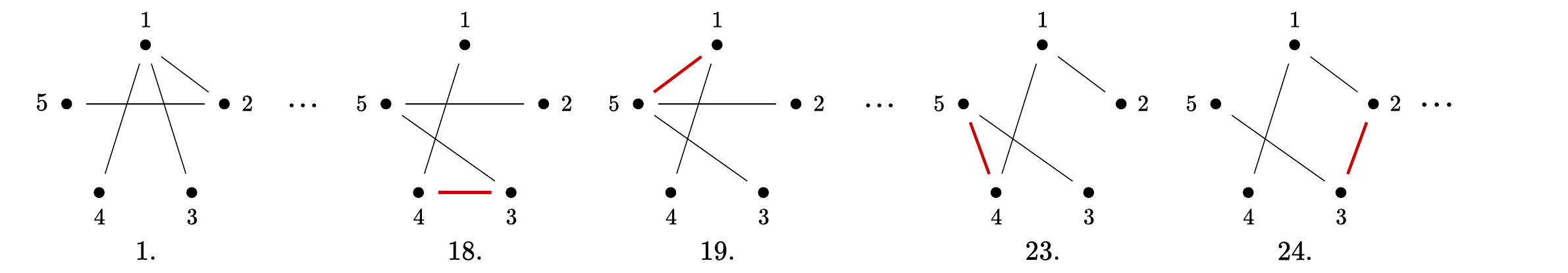}
    \caption{Edge-exchange Gray code for spanning trees of $K_5$ by Smith's algorithm (Algorithm S in \cite{knuth2011}). Only part of the listing is shown. }
    \label{fig:smith}
\end{figure}

\noindent
Interestingly, edge-exchange Gray codes of a graph $G$ correspond to Hamilton paths on the 0/1-polytope that is obtained
as the convex hull of the characteristic vectors of all spanning trees of $G$~\cite{Schrijver2003}.
Among these algorithms, Smith’s algorithm is the only one that generates the Gray code listing in constant amortized time per tree.
The algorithms of Merino, M\"{u}tze, and Williams~\cite{Merino2022} and Merino and M\"{u}tze~\cite{doi:10.1137/23M1612019} generate edge-exchange Gray codes in $O(m\log n(\log \log n)^3)$ and $O(m \log n)$ time per tree, respectively.

Cameron, Grubb and Sawada~\cite{cameron2024} proposed the concept of pivot Gray code (also known as a \emph{strong revolving-door Gray code}~\cite{knuth2011}) for spanning trees of graphs, which imposes stricter constraints on the transitions allowed in Gray codes.
A \emph{pivot Gray code} for the spanning trees of a graph is an edge-exchange Gray code with the additional requirement that the removed and added edges share a common vertex. 
For example, the edge-exchange Gray code generated by Smith’s algorithm is not a pivot Gray code, as the transitions from tree 18 to tree 19, and the transitions from tree 23 to tree 24 in Figure~\ref{fig:smith} for $K_5$ generated by their algorithm lack the pivot property.

Cameron et al.~\cite{cameron2024} developed a greedy strategy to list the spanning trees of a fan graph in a pivot Gray code order. 
For example, the fan graph $F_4$ with four vertices and its pivot Gray code, generated by Cameron et al.'s algorithm, are presented in \cref{fig:fan}. 
The edge-exchange Gray codes provided by Merino, M{\"u}tze, and Williams~\cite{Merino2022} and Merino and M{\"u}tze~\cite{doi:10.1137/23M1612019} are not pivot Gray codes for complete graphs.
\begin{figure}[h]
    \centering
    \includegraphics[width=0.7\linewidth, trim={0 0 0 0}, clip]{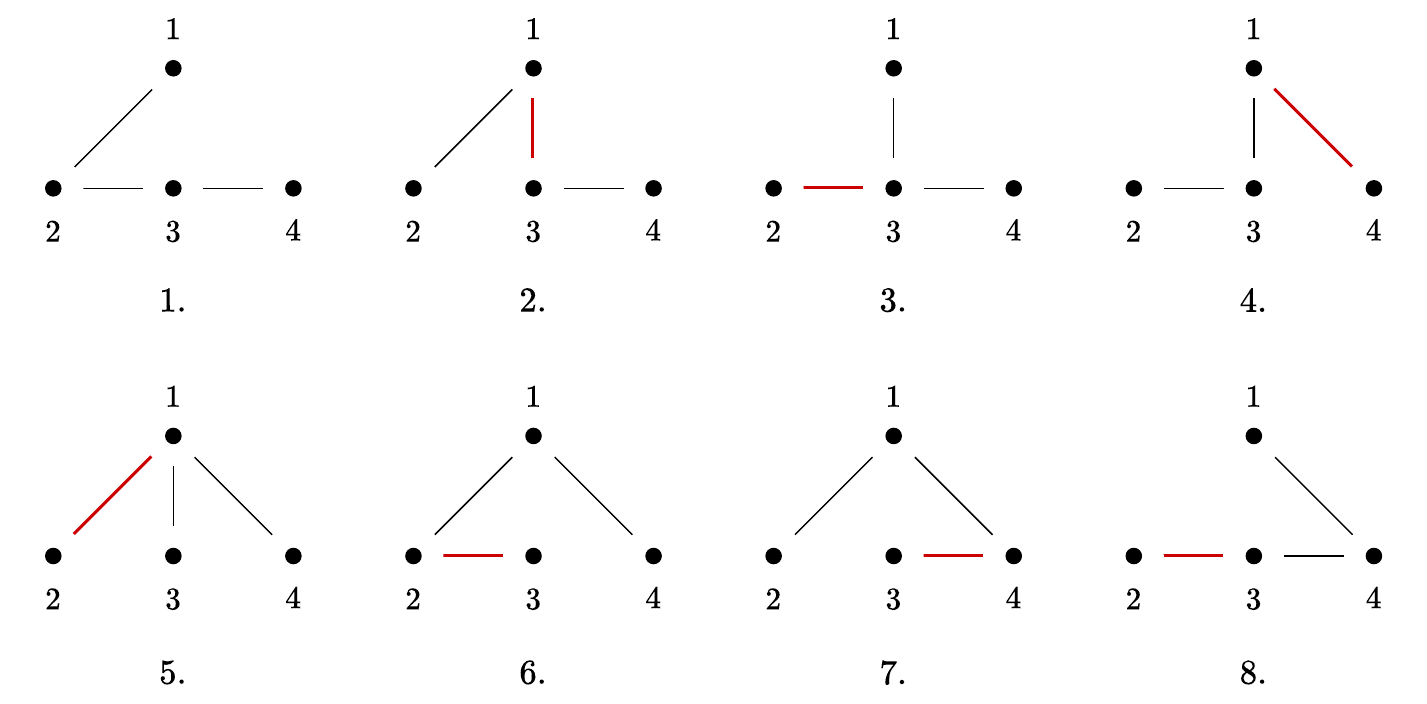}
    \caption{Pivot Gray code for spanning trees of $F_4$.}
    \label{fig:fan}
\end{figure}
The algorithm by Cameron et al. can be optimized to generate the pivot Gray codes for fan graphs in constant amortized time per tree.
They also generalized the algorithm to produce pivot Gray codes for wheel graphs.
More recently, Behrooznia and M\"{u}tze~\cite{behrooznia2025} provided an algorithm that generates pivot Gray codes for outerplanar graphs (which include fan graphs) in $O(n \log n)$ time per tree.
Cameron et al., and Behrooznia and M\"{u}tze posed the following open problems regarding complete graphs:
\begin{itemize}
    \item It remains an open question to find similar results (pivot Gray codes) for other classes of graphs, including the $n$-cube and the complete graph~\cite{cameron2024}.
    \item Knuth~\cite{knuth2011} asked in Exercise 101 in Section 7.2.1.6 whether the complete graph $K_n$ admits a ``nice'' Gray code listing of its spanning trees. One interpretation of this question, in the spirit of Problem 1, could be: Is there a pivot-exchange Gray code for the spanning trees of $K_n$? Alternatively, is there a Gray code listing that provides a more fine-grained explanation of Cayley's formula $n^{n-2}$ for the number of spanning trees~\cite{behrooznia2025}?
\end{itemize}
In this paper, we address these open problems and Knuth's open problem by providing the first known pivot Gray code for spanning trees of complete graphs.
We also provide a simpler proof of Cayley's formula for the number of spanning trees of complete graphs. 
Our algorithm generates each spanning tree in constant amortized time using $O(n^2)$ space.
We further extend our algorithm to generate edge-exchange Gray codes for general graphs with $n$ vertices, achieving $O(n^2)$ time per tree using $O(n^2)$ space.
The runtime can be optimized for specific classes of graphs, including complete bipartite graphs, fan graphs, and wheel graphs, which can be generated in constant amortized time, $O(n)$-amortized time, and $O(n)$-amortized time per tree, respectively, using $O(n^2)$ space.

The rest of the paper is outlined as follows.
In Section~\ref{sec:pivot}, we describe a simple recursive algorithm for generating pivot Gray codes for spanning trees of complete graphs. 
Then in Section~\ref{sec:algorithm}, we generalize our algorithm to generate edge-exchange Gray codes for general graphs. 
We provide some final remarks, including lemmas on related problems, in Section~\ref{sec:remark}. 

\section{Pivot Gray codes for spanning trees of complete graphs.}
\label{sec:pivot}

In this section, we present an algorithm to generate the first known pivot Gray code for spanning trees of $K_n$.
We first introduce a novel encoding for spanning trees of a graph.
We then describe a recursive algorithm to generate pivot Gray codes for spanning trees of $K_n$, which relies on generating Gray codes for $k$-ary strings.
Finally, we present a simple algorithm to efficiently generate a Gray code for $k$-ary strings by leveraging an algorithm for generating Gray codes for reflectable languages.

\subsection{Encoding spanning trees in a graph.}
In our proposed representation for spanning trees of a graph, each spanning tree is represented by an $h$-tuple $T = (t_1, t_2, \ldots, t_h)$, where $h$ is the height of the rooted tree corresponding to the spanning tree rooted at vertex 1. 
The $\ell$-th element $t_\ell$ of $T$ encodes the connections between vertices at level $\ell$ and those at level $\ell-1$ of the rooted tree.

Suppose the rooted tree from level 0 to level $\ell-1$ contains $q_{\ell-1}$ vertices, with $p$ vertices at level $\ell-1$. 
The $\ell$-th element $t_\ell$ of $T$ is a string of length $n-q_{\ell-1}$, that is $t_\ell = a_1 a_2 \cdots a_{n-q_{\ell-1}}$, where each character $a_i$ ranges from 0 to $p$. 
Each character $a_i$ indicates whether the $i$-th smallest remaining vertex (from the $n-q_{\ell-1}$ vertices not yet assigned to levels 0 through $\ell-1$) connects to a vertex at level $\ell-1$.
If $a_i = 0$, then this vertex is not connected to any vertex at level $\ell-1$. 
Otherwise if $a_i \neq 0$, then it connects to the $(a_i)$-th smallest vertex at level $\ell-1$.
For example, the tuple $(00100, 1100, 22)$ corresponds to the spanning tree of $K_6$ with its associated rooted tree shown in Figure~\ref{fig:k6_spanning}.
\begin{figure}[t]
\centering
\includegraphics[width=0.6\linewidth]{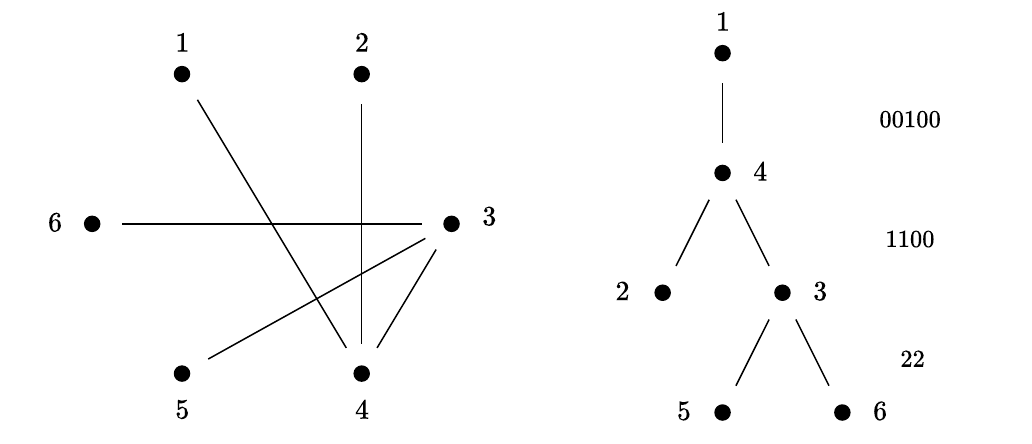}
\label{fig:k6_spanning}
\caption{The spanning tree and rooted tree of $K_6$ that corresponds to the tuple $T = (00100, 1100, 22)$. }
\end{figure}
\noindent Note that in order to have a valid spanning tree, each vertex at level $\ell > 0$ must have a parent and thus $t_\ell \neq 0^{n-q_{\ell-1}}$.

\subsection{Pivot Gray code construction.}\label{sec2:pivotGenS}

Our recursive algorithm leverages Gray codes for $k$-ary strings to enumerate all spanning trees of $K_n$. We first provide an overview of the algorithm, with details and optimizations described in the next subsection. The general approach is to generate all possible tuples $T$ corresponding to the spanning trees of $K_n$, while minimizing differences between consecutive tuples to ensure a pivot Gray code for the spanning trees.

To achieve this, our algorithm recursively generates all possible rooted trees rooted at vertex 1 at level zero, with a fixed subtree $S_{\ell-1}$ spanning levels 0 to $\ell-1$. 
Each generated rooted tree includes $S_{\ell-1}$ as its subtree and extends it by connecting the remaining vertices (those not in $S_{\ell-1}$) to vertices at level $\ell-1$ in all valid configurations, ensuring that consecutive trees differ by a pivot change. 
Initially, $S_0$ only contains vertex 1, and the algorithm generates all rooted trees containing vertices $\{2, 3, \ldots, n\}$ with $S_0$ as the subtree at level 0, maintaining the pivot Gray code property. 
The algorithm recursively applies this process until all rooted trees with a given subtree $S_{\ell-1}$ are exhaustively generated at level $\ell$. 
It then backtracks to level $\ell-1$, generates the next configuration of $S_{\ell-1}$ with the same subtree $S_{\ell-2}$, and recursively generates all rooted trees with the updated $S_{\ell-1}$ at level $\ell$. 
This continues until all configurations at each level are exhausted.

Our algorithm begins with the spanning tree $1-2-\cdots-n$. 
As an example, Figure~\ref{fig:kn_GrayR} illustrates a pivot Gray code listing for $K_4$, starting with the spanning tree $1-2-3-4$. 
At level three of a rooted tree rooted at vertex 1, only one configuration exists for the given subtree $S_2$ (e.g., vertex 4 connected to vertex 3). 
The algorithm backtracks to level two, where two configurations with the same subtree $S_1$ have not yet appeared (vertex 2 connected to both vertices 3 and 4, and vertex 2 connected to vertex 4). 
A pivot change is applied (for example, reassigning the parent of vertex 4 from vertex 3 to vertex 2), and the algorithm attempts to generate all configurations at level three. 
For the configuration with vertex 2 as the parent of vertices 3 and 4, no vertices remain at level three, so no further configurations are generated. Then, for the configuration with vertex 2 as the parent of vertex 4, level three has one configuration (vertex 3 connected to vertex 4). 
After exhausting all configurations with the given $S_1$ and $S_2$ at levels two and three, the algorithm backtracks to level one, generates the next configuration of $S_1$, and recursively generates all rooted trees with the updated $S_1$ at level two and subsequently level three. 
This process continues until all spanning trees of $K_4$ are generated, as shown in Figure~\ref{fig:kn_GrayR}.

\begin{figure}[t]
\centering
\includegraphics[width=1\linewidth, trim={0 0 25 0}, clip]{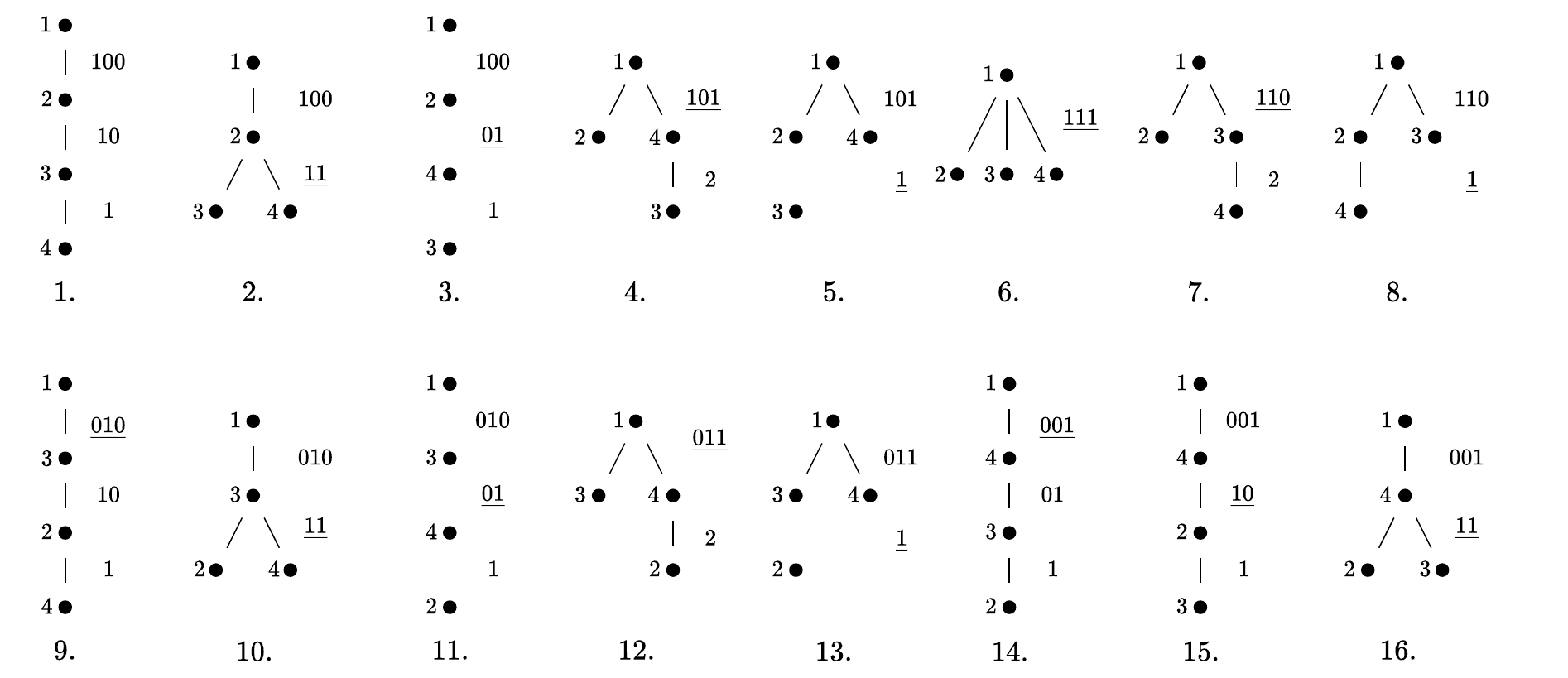}
    \caption{Pivot Gray code for spanning trees of $K_4$.}
\label{fig:kn_GrayR}
\end{figure}

The question then becomes: How can we ensure that repeatedly applying pivot changes generates all possible configurations. We first prove the following lemmas.

\begin{lemma}~\label{lem:exhaust}
Let $T = (t_1, t_2, \ldots, t_h)$ be the tuple corresponding to a spanning tree of $K_n$ rooted at vertex 1, with $p$ vertices at level $\ell-1$ and $q_{\ell-1}$ vertices across levels 0 to $\ell-1$.
To exhaustively generate all spanning trees of $K_n$ rooted at vertex 1 that share the same subtree $S_{\ell-1}$ spanning levels 0 to $\ell-1$ (i.e., those whose corresponding tuples are of the form $T' = (t_1, t_2, \ldots, t_{\ell-1}, t_\ell', t_{\ell+1}', \ldots, t_h')$) exactly once, it is necessary to list every $(p+1)$-ary string of length $n-q_{\ell-1}$ exactly once for $t_\ell'$, excluding the all-zero string $0^{n-q_{\ell-1}}$.
\end{lemma}

\begin{proof}
If $q_{\ell-1}>0$, then since the spanning tree is connected, at least one vertex at level $\ell$ must connect to a vertex at level $\ell-1$ and thus $t_\ell' \neq 0^{n-q_{\ell-1}}$. Moreover, since $K_n$ is a complete graph, all possible connections from the $n-q_{\ell-1}$ vertices not in levels 0 to $\ell-1$ to the $p$ vertices at level $\ell-1$ are valid. 
These possible connections correspond to all possible $(p+1)$-ary strings of length $n-q_{\ell-1}$ for $t_\ell'$ except $0^{n-q_{\ell-1}}$.
\end{proof}


\begin{lemma}~\label{lem:diff}
Let $T = (t_1, t_2, \ldots, t_h)$ be the tuple corresponding to a spanning tree of $K_n$ rooted at vertex 1 with $q_{\ell-1}$ vertices across levels 0 to $\ell-1$.
Suppose $t_\ell'$ is obtained by either changing a symbol of $t_\ell$, or by swapping a 0 and an $r > 0$ of $t_\ell$ when $t_\ell = 0^{n-q_{\ell-1}-b-1}r0^b$ for some $b \ge 0$.
Then there exists a tuple $T' = (t_1, t_2, \ldots, t_{\ell-1}, t_\ell', t_{\ell+1}', \ldots, t_{j}')$ corresponding to a spanning tree of $K_n$ rooted at vertex 1 that is obtainable from $T$ by a single pivot change.
\end{lemma}

\begin{proof}
Let $t_\ell = a_1 a_2 \cdots a_{n-q_{\ell-1}}$, where $a_v$ in $t_\ell$ corresponds to the connection of a vertex $v$  to a vertex at level $\ell-1$, where $v$ is not in $S_{\ell-1}$. We consider the following four cases for a pivot change on an edge involving a vertex at level $\ell-1$, while preserving the subtree $S_{\ell-1}$ spanning levels 0 to $\ell-1$:

\begin{enumerate}[label=(\alph*)]
\item $v$ at level $q > \ell$ disconnects from its parent $w$ and connects to a vertex $u$ at level $\ell-1$: Initially, $a_v = 0$, and it updates to $a_v = r$ for some $r > 0$, requiring a single symbol change in $t_\ell$;
\item $v$ at level $\ell$ disconnects from its parent $u$ at level $\ell-1$ and connects to a vertex $w$ at level $q \geq \ell$: Initially, $a_v = r$ for some $r > 0$, and it updates to $a_v = 0$, requiring a single symbol change in $t_\ell$;
\item $v$ at level $\ell$ disconnects from its parent $u$ at level $\ell-1$ and connects to another vertex $w$ at level $\ell-1$: Initially, $a_v = r$ for some $r > 0$, and it updates to $a_v = j$ for some $j > 0$ where $j \neq r$, requiring a single symbol change in $t_\ell$;
\item A vertex $u$ at level $\ell-1$ disconnects from its child $v$ at level $\ell$ and connects to a vertex $w \neq v$ at level $q > \ell$, with $v$ being the only vertex at level $\ell$: Let $a_w$ be the character in $t_\ell$ that corresponds to the connection of vertex $w$ to a vertex at level $\ell-1$.
Initially, $a_v = r$ for some $r > 0$, and $a_w = 0$ since $w$ is at level $q > \ell$; they update to $a_v = 0$ and $a_w = r$, requiring a swap between a 0 and a non-zero character in $t_\ell$. To maintain the connectivity, $w$ must be a descendant of $v$. 
\end{enumerate}
These four cases are illustrated in Figure~\ref{fig:kn_Gray}.
\begin{figure}[t]
    \centering
    \includegraphics[width=0.95\linewidth, trim={25 0 0 0}, clip]{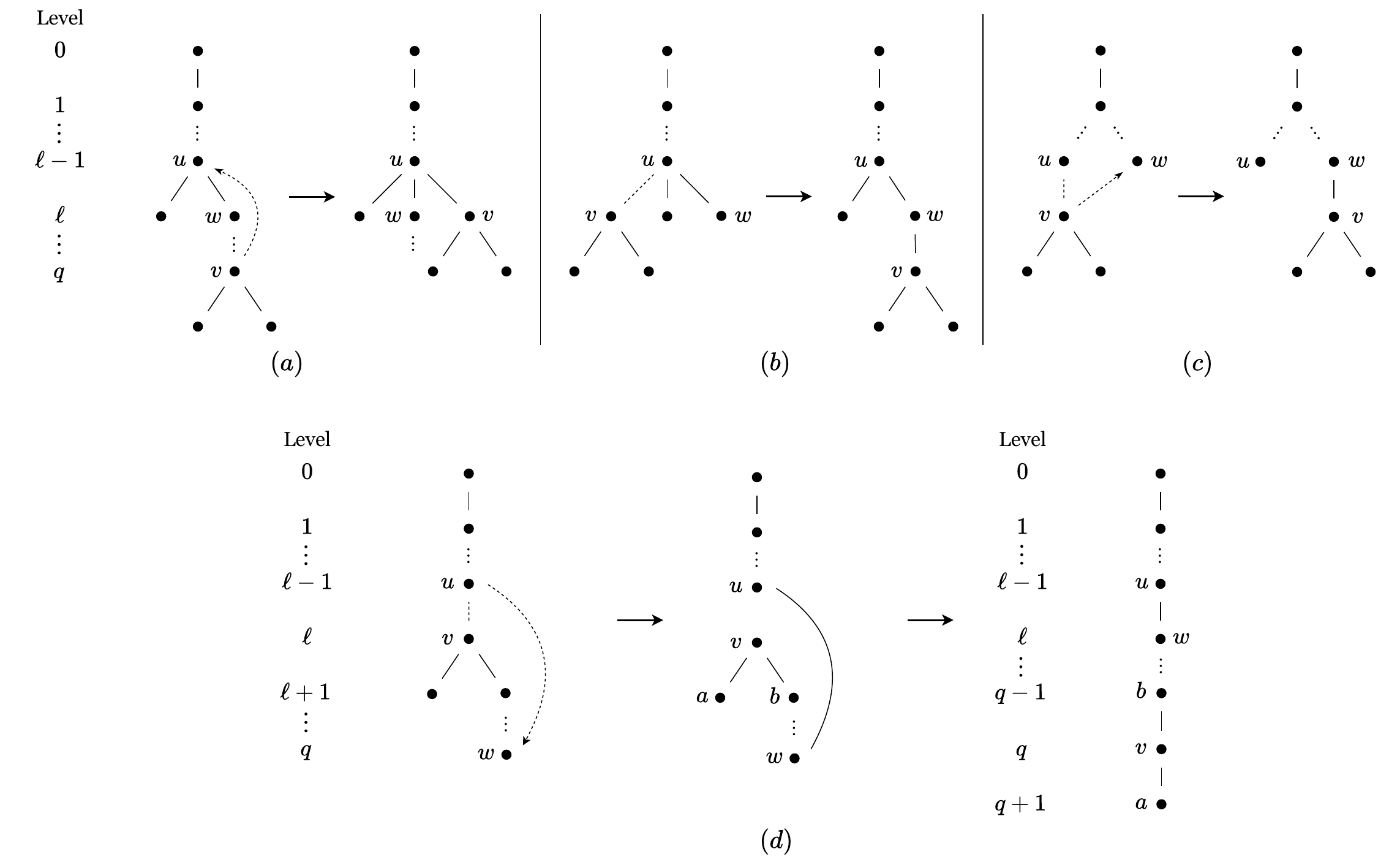}
    \caption{The four cases for a pivot change
on an edge involving a vertex at level $\ell - 1$, while preserving the subtree $S_{\ell-1}$ spanning levels 0 to $\ell - 1$.}
    \label{fig:kn_Gray}
\end{figure}
Each case corresponds to applying a pivot change to an edge involving a vertex at level $\ell-1$. 
The first three cases (cases (a), (b), and (c)) correspond to a single symbol change in $t_\ell$, and the last case (case (d)) corresponds to the swapping of a 0 and an $r > 0$ in $t_\ell$ when $t_\ell = 0^{n-q_{\ell-1}-b-1}r0^b$ for some $b \ge 0$.
\end{proof}

As previously described, our recursive algorithm generates pivot Gray codes for $K_n$ by recursively enumerating all possible rooted trees rooted at vertex 1, with a given subtree $S_{\ell-1}$ fixed for levels 0 to $\ell-1$. 
By Lemma~\ref{lem:exhaust}, this corresponds to generating all possible $(p+1)$-ary strings for $t_\ell$ in the tuple $T = (t_1, t_2, \ldots, t_h)$, excluding $0^{n-q_{\ell-1}}$, while keeping $t_1$ to $t_{\ell-1}$ unchanged. 
Thus, the algorithm can be interpreted as traversing listings of $(p+1)$-ary strings of length $n-q_{\ell-1}$ that satisfy the following conditions:
\begin{enumerate}
\item Consecutive strings differ by a single symbol change, or by swapping a 0 and a non-zero character, which both strings have exactly one non-zero character (Lemma~\ref{lem:diff});
\item 
The listing exhaustively includes all possible $(p+1)$-ary strings of length $n-q_{\ell-1}$. The listing must be able to start from any valid $(p+1)$-ary string. This is necessary because a pivot change at level $\ell$ may require strings at levels greater than $\ell$ to begin with an arbitrary $(p+1)$-ary string (see tree 14 in Figure~\ref{fig:kn_GrayR});
\item The string $0^{n-q_{\ell-1}}$ is excluded.
\end{enumerate}

Such listings can be produced using Gray codes for $k$-ary strings. 
There are many algorithms for generating 1-Gray codes for the set of $k$-ary strings~\cite{mutze2023combinatorial, ruskey1996combinatorial, Sav97, 10.1007/978-3-030-85088-3_15,flipswap,Vajnovszki2007GrayCO,williams2013greedy}. 
With minor modifications, many of these algorithms can be adapted to meet the above conditions. 
For example, consider the following cyclic 1-Gray code listings $L^3_2$, $L^2_2$, $L^1_3$, and $L^1_2$ for the sets of binary strings of length three, binary strings of length two, ternary strings of length one, and binary strings of length one, respectively:
\begin{itemize}
\item $L^3_2$: 100, 101, 111, 110, 010, 011, 001, 000;
\item $L^2_2$: 10, 11, 01, 00;
\item $L^1_3$: 2, 1, 0;
\item $L^1_2$: 1, 0.
\end{itemize}
Let $L'$ denote the sequence of strings of length $j$ obtained by removing the string $0^{j}$ from a listing $L$. 
By removing the string $0^{j}$ from each of the above listings, we obtain the following sequences that satisfy the conditions:
\begin{itemize}
\item ${L^3_2}'$: 100, 101, 111, 110, 010, 011, 001;
\item ${L^2_2}'$: 10, 11, 01;
\item ${L^1_3}'$: 2, 1;
\item ${L^1_2}'$: 1.
\end{itemize}
Note that $L^3_2$, $L^2_2$, $L^1_3$, and $L^1_2$ are cyclic, meaning the first and last strings also differ by a single symbol change or a swap between a 0 and a non-zero character.

We now demonstrate how our algorithm uses the listings ${L^3_2}'$, ${L^2_2}'$, ${L^1_3}'$, and ${L^1_2}'$ to construct a pivot Gray code for the spanning trees of $K_4$. Initially, vertex 1 is at level zero of the rooted tree, and the algorithm recursively enumerates all possible strings for $t_\ell$ in the tuple $T = (t_1, t_2, \ldots, t_h)$. For the connection between levels zero and one, since level zero only contains vertex 1, then $n-q_{0} = n-1 = 3$, and thus $t_1$ is a binary string of length three. 
Thus, our algorithm enumerates each string in ${L^3_2}'$. 
The first string in ${L^3_2}'$ is 100, so $t_1 = 100$, which corresponds to vertex 2 at level one connecting to vertex 1 at level zero.
Next, the algorithm recursively considers level two, where $n-q_{1} = n-2 = 2$ since the remaining vertices are vertices 3 and 4, and level one contains only vertex 2. 
Thus, $t_2$ is a binary string of length two, and the algorithm enumerates each string in ${L^2_2}'$. 
The first string in ${L^2_2}'$ is 10, so $t_2 = 10$, which corresponds to vertex 3 at level two connecting to vertex 2 at level one. 
The algorithm then recursively considers level three, with the only remaining vertex 4, and level two containing only vertex 3. 
Thus, $t_3$ is a binary string of length one, and the algorithm enumerates each string in ${L^1_2}'$. 
The only string in ${L^1_2}'$ is 1, so $t_3 = 1$ which corresponds to vertex 4 at level three connecting to vertex 3 at level two.
Since ${L^1_2}'$ contains only one string, all strings in ${L^1_2}'$ are exhausted. 
The algorithm then backtracks to level two and generates the next string in ${L^2_2}'$ after 10, which is 11. 
This implies that vertex 2 at level one is connected to both vertices 3 and 4 at level two. 
This process continues and the resulting listing of spanning trees for $K_4$ is shown in Figure~\ref{fig:kn_GrayR}.
Note that in some cases, the algorithm may not start with the first string in these sequences. 
For example, in the transition from tree 13 to tree 14 in Figure~\ref{fig:kn_GrayR}, the pivot change disconnects the parent of vertex 3 and connects vertex 3 to vertex 4. 
At level two of tree 14, the connection between vertices 4 and 3 corresponds to $t_2 = 01$, which is not the first string in ${L^2_2}'$. 
In such cases, the algorithm starts with 01 in ${L^2_2}'$ and continues generating strings until it reaches the end of the sequence, then cycles to the front of the sequence to exhaust all strings in ${L^2_2}'$. 
Since ${L^2_2}'$ is cyclic, consecutive strings differ by a single symbol change or by swapping a 0 and a non-zero character, ensuring the resulting listing remains a Gray code.

\begin{algorithm}[t]
    \caption{Generating pivot Gray codes for spanning trees of complete graphs}
    \label{alg:spanning_complete}
    \begin{algorithmic}[1]
        \Procedure{\sc GenS}{$V, t_{\ell-1} = a_1 a_2 \cdots a_{|V|}$}
        \State{$P \gets v_i \in V \mathrm{~where~} a_i > 0$}
        \State{$C \gets v_i \in V \mathrm{~where~} a_i = 0$}
        \If{$\vert C \vert = 0$}
            \State{Print spanning tree}
            \State{\textbf{return}}
        \EndIf
        \State{$t_{\ell} = b_1 b_2 \cdots b_{|C|} \gets 0^{\vert C\vert}$}
        \For{$c_i \in C$}
        \If{$\pi(c_i) \in P$}{$~b_i \gets j $~where~$\pi(c_i) = p_j \in P$}
        \EndIf
        \EndFor
        \State{{\sc GenS}($C, t_{\ell}$)}
        
        \While{{\sc NextK}$(t_{\ell}, \vert P\vert +1) \ne \emptyset$}
            \State{$t_{\ell}, \delta \gets ${\sc NextK}$(t_{\ell}, \vert P \vert+1)$}
            \If{$\vert \delta \vert > 1$}
                \State{$\pi(c_{\delta_1}) \gets \pi(c_{\delta_0})${\color{blue}$\quad \quad \quad  \quad  \quad  \quad  \quad \quad \hspace{0.5em} \triangleright  \ $case (d)}}
                \State{{\sc Lift}$(c_{\delta_1})$}
            \ElsIf{$b_{\delta_0} > 0$}
                \State{$\pi(c_{\delta_0}) \gets p_{b_{\delta_0}} \in P$ {\color{blue}$\quad \quad  \quad  \quad  \quad \quad  \quad  \hspace{0.25em}\triangleright \  $case (a) and (c)}}
            \Else
                \State{$\pi(c_{\delta_0}) \gets c_j \in C$ where $b_j > 0$ {\color{blue}$\quad  \quad \hspace{0.1em}\triangleright \  $case (b)}}
            \EndIf
            \State{{\sc GenS}($C, t_{\ell}$)}
        \EndWhile
        \EndProcedure
    \end{algorithmic}
\end{algorithm}

Our algorithm {\sc GenS} generates a pivot Gray code for spanning trees of $K_n$ takes two parameters, where $V$ is the set of vertices not yet included in the rooted tree up to level $\ell-2$, and $t_{\ell-1} = a_1 a_2 \cdots a_{|V|}$ is a $k$-ary string encoding parent assignments for vertices at level $\ell-1$. 
A parent pointer is maintained for each vertex $v$, where $\pi(v)$ denotes its parent in the rooted tree. 
The algorithm initializes with the spanning tree $1-2-\cdots-n$, setting $\pi(v) = v-1$ for each vertex $v \geq 2$ and $\pi(1) = 0$ (indicating the root).
The algorithm maintains two sets: $P = \{v_i \in V \mid a_i > 0\}$, containing vertices at level $\ell-1$, and $C = \{v_i \in V \mid a_i = 0\}$, containing vertices not in levels 0 to $\ell-1$. 
If $|C| = 0$, then the rooted tree includes all $n$ vertices of $K_n$, so the algorithm prints the spanning tree and returns. 
Otherwise, our algorithm updates $t_\ell = b_1 b_2 \cdots b_{|C|}$ by checking the current connections between vertices in $P$ and vertices in $C$.
It then calls {\sc GenS}$(C, t_\ell)$ to recursively construct rooted trees with the same subtree $S_{\ell-1}$.
Next, the algorithm iteratively updates $t_\ell$ using {\sc NextK}$(t_\ell, |P|+1)$, which generates the next $k$-ary string in Gray code order with base $k = |P|+1$. 
The function {\sc NextK} returns the updated $t_\ell$ and the indices $\delta$ which stores the positions changed from the previous string.
Efficient generation of $k$-ary Gray codes is discussed in Section~\ref{sec:k-ary}. 
If $|\delta| > 1$, then this corresponds to case (d) of Lemma~\ref{lem:diff}, the algorithm sets
the parent of $c_{\delta_1} \in C$ to the parent of $c_{\delta_0} \in C$. 
We then ``lift'' the vertex $c_{\delta_1} \in C$ to level $\ell$ and updating the parent pointers of ancestors of $c_{\delta_1}$(line 15, Algorithm~\ref{alg:spanning_complete} by calling {\sc Lift}$(c_{\delta_1})$). 
If $|\delta| = 1$, there are two cases:
\begin{itemize}
\item If $b_{\delta_0}$ is updated to a value $r > 0$, then this case corresponds to cases (a) or (c) of Lemma~\ref{lem:diff}, the algorithm sets the parent of $c_{\delta_0}$ to $p_r \in P$;
\item If $b_{\delta_0}$ is updated to 0, then this case corresponds to case (b) of Lemma~\ref{lem:diff}, the algorithm sets the parent of $c_{\delta_0}$ to a vertex $c_j \in C$ where $b_j > 0$ (a sibling at level $\ell$ with a parent).
\end{itemize}
The algorithm then calls {\sc GenS}$(C, t_\ell)$ to recursively construct rooted trees with the updated subtree $S_\ell$. 
Pseudocode of the algorithm is shown in Algorithm~\ref{alg:spanning_complete}. 
To run the algorithm, we make the initial call {\sc GenS}$(\{1, 2, 3, \ldots, n\}, (10^{n-1}))$ which sets $V = \{1, 2, 3, \ldots, n\}$, and $t_0 = 10^{n-1}$. 
A complete Python implementation of the algorithm, incorporating optimizations discussed in later subsections, is provided in Appendix \ref{sec2:gens_opt_code}.

\begin{theorem}~\label{thm:GrayKN}
The algorithm {\sc GenS} generates a pivot Gray code for the spanning trees of $K_n$.
\end{theorem}

\begin{proof}
By Lemma~\ref{lem:diff}, the algorithm {\sc GenS} ensures that consecutive spanning trees differ by a single pivot change. To establish the pivot Gray code property, it suffices to show that each spanning tree of $K_n$ appears exactly once in the listing. We prove this by contradiction.

Suppose there exists a spanning tree of $K_n$ that does not appear in the listing generated by {\sc GenS}. Such a spanning tree can be encoded as a tuple $T = (t_1, t_2, \ldots, t_h)$. By Lemma~\ref{lem:exhaust}, the absence of $T$ in the listing implies that no spanning tree in the listing has a subtree $S_{h-1}$, spanning levels 0 to $h-1$, that matches the first $h-1$ components $(t_1, t_2, \ldots, t_{h-1})$ of $T$. Otherwise, {\sc GenS} would exhaustively generate a Gray code for $t_h$ (via {\sc GenK}) to include $T$. Applying the same argument recursively, no spanning tree in the listing has a subtree $S_{h-2}$, spanning levels 0 to $h-2$, matching $(t_1, t_2, \ldots, t_{h-2})$. Continuing applying the same reasoning implies that no spanning tree has a subtree $S_0$ at level zero matching $t_1$. However, at level zero, {\sc GenS} generates a Gray code of length $n-1$ (via {\sc GenK}) that includes all binary strings except $0^{n-1}$, covering all possible $t_1$ for spanning trees of $K_n$, a contradiction.
\end{proof}

Storing the entire Gray code listing for $k$-ary strings, however, requires exponential space. 
In the next subsection, we present a simple algorithm that generates the next $k$-ary string in the Gray code listing using only a linear amount of space.

\subsection{Efficient generation of the underlying Gray codes for \texorpdfstring{$k$}{k}-ary strings and the pivot Gray codes for spanning trees of \texorpdfstring{$K_n$}{Kn}.}\label{sec:k-ary}

To generate a Gray code for $k$-ary strings, excluding the string $0^{n-q_{\ell-1}}$, that satisfies the three conditions discussed in Section \ref{sec2:pivotGenS}, we adopt the reflectable language framework described in \cite{li2009gray}. 
In~\cite{li2009gray}, Li and Sawada developed a framework for generating Gray codes for reflectable languages. 
A language $L$ over an alphabet $\Sigma$ is reflectable if, for every $i > 1$, there exist two symbols $x_i, y_i \in \Sigma$ such that if $w_1 w_2 \cdots w_{i-1}$ is a prefix of a word in $L$, then both $w_1 w_2 \cdots w_{i-1} x_i$ and $w_1 w_2 \cdots w_{i-1} y_i$ are also prefixes of words in $L$. 
By reflecting the order of the children and using the special symbols $x_i$ and $y_i$ as the first and last children of each node at level $i-1$, Li and Sawada devised a generic recursive algorithm to generate Gray codes for reflectable languages, including combinatorial objects such as $k$-ary strings, restricted growth functions, and $k$-ary trees. 
For example, their algorithm generates the binary reflected Gray code discussed in Section~\ref{sec:intro} when setting $x_i = 0$ and $y_i = 1$ for every level $i > 1$.

\begin{lemma}
The set of $k$-ary strings of length $j$, excluding the string $0^j$, is a reflectable language when $k > 2$.
\end{lemma}

\begin{proof}
By setting $x_i = k-1$ and $y_i = k-2$ for each level $i$, $k$-ary strings of length $j$ are reflectable.
\end{proof}

Since the set of $k$-ary strings of length $j$, excluding the string $0^j$, is a reflectable language, we can leverage the reflectable language framework described in~\cite{li2009gray} to generate a 1-Gray code for this set in constant amortized time per string using only a linear amount of space. 
The framework is also quite versatile, allowing the Gray code to start with any arbitrary string. 
This is particularly useful in the construction of our pivot Gray code for spanning trees of $K_n$, where by the second condition in Section~\ref{sec2:pivotGenS}, a pivot change at level $\ell$ may require strings at levels greater than $\ell$ to begin with an arbitrary $k$-ary string of length $j$.

\begin{figure}[t]
\centering
\includegraphics[width=1\linewidth,trim={10 0 30 0}, clip]{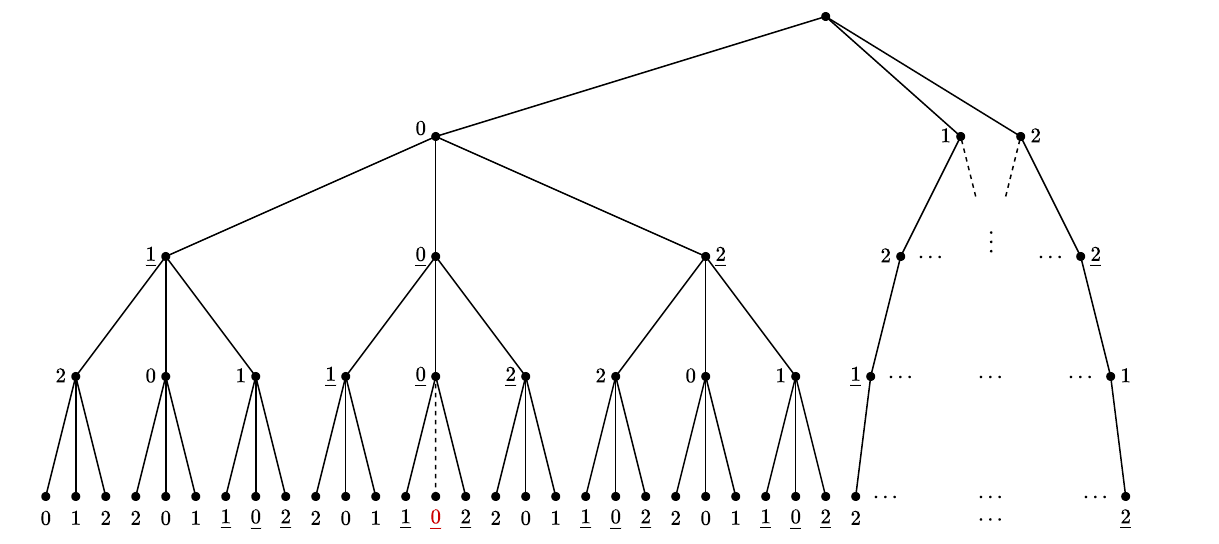}
    \caption{Reflectable Gray code of ternary strings of length four generated by {\sc GenK}, starting with 0120. }
\label{fig:reflectable}
\end{figure}

\begin{algorithm}[t]
    \caption{Generating $k$-ary strings starting with $\alpha$ in reflectable Gray code order.}
    \label{alg:k-ary}
    \begin{algorithmic}[1]
        \Procedure{\sc GenK}{$i, k, \alpha = a_1 a_2 \cdots a_j$}
        \For{$r = 1 $~to~$k$}
        \If{$i = \vert \alpha \vert$}
        \If{$s \ne 0$}
        \State {Output $\alpha, \delta$}
        \State {$\delta \gets \emptyset$}
        \EndIf
        
        \Else
        \State{{\sc GenK}$(i+1, k, \alpha)$}
        \If{$a_{i+1} = k-1$}{~$d_{i+1} \gets 1$}
        \Else{~$d_{i+1} \gets -d_{i+1}$}
        \EndIf
        \EndIf

        \If{$r < k$}
        \State{$s \gets s + \big((a_i + d_i + k) \mod k\big) - a_i$}
        \State{$a_i \gets (a_i + d_i + k) \mod k$}
        \State{$\delta \gets \delta, i$}
        \EndIf
        \EndFor
        \EndProcedure
    \end{algorithmic}
\end{algorithm}

Our algorithm for generating a 1-Gray code for $k$-ary strings of length $j$, excluding $0^j$, maintains a global direction array $d$, where $d_i$ indicates the ordering direction for the $i$-th character. 
If $d_i = 1$, the children of the $i$-th character start with $a_i$ and increment by 1 (modulo $k$) until reaching $a_i-1$. 
If $d_i = -1$, the children start with $a_i-1$ and decrement by 1 (modulo $k$) until reaching $a_i$. 
The algorithm avoids generating the string $0^j$. Two special cases ensure the listing remains a 1-Gray code when $k > 2$ by preventing $x_i = 0$ or $y_i = 0$, which would disrupt the reflectable property if $0^j$ were removed:
\begin{itemize}
\item If the $i$-th character of the initial string $\alpha = a_1 a_2 \cdots a_{j}$ has $a_i = 1$, we set $d_i = -1$, so its children start with $y_i = 1$ and end with $x_i = 2$;
\item If the $i$-th character of the initial string $\alpha = a_1 a_2 \cdots a_{j}$ has $a_i = 0$, we set $d_i = 1$ for the first $2k$ characters, updating $a_i$ from 0 to $k-1$, and then from $k-1$ to $k-2$ for subsequent strings. 
We then set $x_i = k-1$ and $y_i = k-2$, and apply the usual reflectable strategy for the remaining updates of $a_i$.
\end{itemize}

As an example, our algorithm generates a 1-Gray code for ternary strings of length four excluding the string 0000, starting with the initial string 0120 as shown below: 
\begin{center}
0120, 0121, 0122, 0102, 0100, 0101, 0111, 0110, 0112, 0012,  0010, 0011, 0001, 0002, 0022, 0020,\\ 0021, 0221, 0220, 0222,  0202, 0200, 0201, 0211, 0210, 0212, 1212, 1210, 1211, 1201,  1200, 1202,\\ 1222, 1220, 1221, 1021, 1020, 1022, 1002, 1000,  1001, 1011, 1010, 1012, 1112,  1110, 1111, 1101,\\ 1100, 1102,  1122, 1120, 1121, 2121, 2120, 2122, 2102, 2100, 2101, 2111,  2110, 2112, 2012, 2010,\\ 2011, 2001, 2000, 2002, 2022, 2020,  2021, 2221, 2220, 2222, 2202,   2200, 2201, 2211, 2210, 2212. 
\end{center}
The recursive construction tree is shown in Figure~\ref{fig:reflectable}. 
Pseudocode of the algorithm is shown in Algorithm~\ref{alg:k-ary}. 
To run the algorithm, we make the initial call {\sc GenK}$(i, k, \alpha)$, where $i$ is the current level of the recursive computation tree (initialized to 1), $k$ is the base of the $k$-ary strings, and $\alpha$ is the initial string of the listing. 
A complete Python implementation of an advanced of {\sc GenK}, namely {\sc GenMR}, which generates Gray codes for mixed-radix strings, is provided in Appendix \ref{sec:genk_code}. 
This implementation can generate Gray codes for $k$-ary strings by setting $k_i = k$ for all positions.

\begin{lemma} \label{lem:genk_gc}
Consecutive strings generated by {\sc GenK} differ by a single symbol change, or by swapping a 0 and a 1 which the strings have only one single non-zero character. 
\end{lemma}
\begin{proof}
If $k > 2$, the set of $k$-ary strings of length $j$, excluding the string $0^j$, forms a reflectable language. 
The algorithm {\sc GenK} selects $x_i$ and $y_i$ as non-zero characters, ensuring the listing is a 1-Gray code by the results in~\cite{li2009gray}. 
Thus, consecutive strings differ by a single symbol change. 
Otherwise, if $k = 2$, observe that the set of binary strings of length $j$, including the string $0^j$, is a reflectable language, and {\sc GenK} generates a 1-Gray code for this set. 
When the string $0^j$ is removed, there are two cases: if $0^j$ is at the end of the listing, its removal does not affect the Gray code property; otherwise, if $0^j$ appears between two strings, say $0^x 1 0^{j-x-1}$ and $0^y 1 0^{j-y-1}$ where $x \neq y$, these strings differ by a swap between a 0 and a 1.
\end{proof}

The algorithm {\sc GenK} efficiently generates pivot Gray codes for spanning trees of $K_n$. 
For example, the algorithm {\sc GenK} produces the following 1-Gray code listings for $k$-ary strings of length $j$, starting with an initial string $\alpha$:
\begin{itemize}
\item $j = 3, k = 2, \alpha = 100$: 100, 101, 111, 110, 010, 011, 001;
\item $j = 2, k = 2, \alpha = 10$: 10, 11, 01;
\item $j = 2, k = 2, \alpha = 11$: 11, 01, 10;
\item $j = 1, k = 3, \alpha = 2$: 2, 1;
\item $j = 1, k = 3, \alpha = 1$: 1, 2.
\end{itemize}
The resulting listing of spanning trees for $K_4$, generated using these underlying Gray codes for $k$-ary strings of length $j$ by {\sc GenS} (Algorithm~\ref{alg:spanning_complete}), is shown in Figure~\ref{fig:kn_GrayR}.

We now prove the complexity of our algorithm {\sc GenS} (Algorithm~\ref{alg:spanning_complete}) for generating pivot Gray codes for spanning trees of $K_n$.
As standard for generation algorithms, the time required to output a spanning tree is not part of the analysis.
\begin{theorem} \label{thm:gens_complexity}
Spanning trees of complete graphs $K_n$ can be generated in pivot Gray code order in constant amortized time per tree, using $O(n^2)$ space.
\end{theorem}
\begin{proof}
The algorithm {\sc GenS} recursively constructs spanning trees of $K_n$ by enumerating $k$-ary strings for $t_\ell$ in the tuple $T = (t_1, t_2, \ldots, t_h)$, where $t_\ell$ encodes connections from level $\ell$ to level $\ell-1$. By Lemma~\ref{lem:genk_gc}, the function {\sc GenK} generates 1-Gray code listings for $k$-ary strings of length $|t_\ell|$, excluding $0^{|t_\ell|}$, in constant amortized time per string, as established in~\cite{li2009gray}. Using a function iterator, {\sc NextK} updates $t_\ell$ in constant amortized time per string by tracking the position that differs from the previous string, which is achieved by storing the last level of the computation tree where the character is not $x_i$ or $y_i$.

The algorithm {\sc GenS} has at most $n-1$ levels in the rooted tree. The space required by {\sc GenK} includes the direction array $d$ of length at most $n$, which uses $O(n)$ space for each level $\ell$, thus utilizing $nO(n) = O(n^2)$ space. For {\sc GenS}, let $t = |t_{\ell-1}|$. Updating sets $P$ and $C$ (lines 2 and 3 of Algorithm~\ref{alg:spanning_complete}) takes $O(t)$ time when $P$ and $C$ are maintained as linked lists. 
Computing $t_\ell = b_1 b_2 \cdots b_{|C|}$ which describes connections between $P$ and $C$, can be optimized to $O(t)$ by maintaining a hash table $m$ that maps $\pi(c_i) = p_j \in P$ to $j$. The details of the hash table $m$ are provided in Appendix \ref{sec:gens_opt}.
Thus, operations before the while loop (lines 2--10 of Algorithm~\ref{alg:spanning_complete}) require $O(t)$ time. 
However, the while loop (lines 11--20 of Algorithm~\ref{alg:spanning_complete}) generates \( (|P| + 1) \cdot |C| - 1 \) configurations. In the worst case, it produces at least \( t - 1 \) configurations when \( |C| = 1 \), which differs by a constant factor; otherwise, the number of configurations grows exponentially and is far larger than $t$. 
This amortizes the \( O(t) \) time to constant time per tree. 

Within the while loop, the lift operation (line 15 of Algorithm~\ref{alg:spanning_complete}) requires $O(t)$ time in the worst case. However, this operation is only triggered by a swap between a 0 and a non-zero character, which occurs at most once per 1-Gray code listing of $2^t$ binary strings ($k=2$). Thus, its cost is amortized to constant time per tree. 
For the operation in line 19 of Algorithm~\ref{alg:spanning_complete}, which connects a disconnected \( c_i \in C \) to its sibling, a hash table \( C' \) with chaining via a doubly linked list can be maintained to efficiently store and access the siblings of $c_i$ at level \( \ell \) in constant time. 
The implementation details are   provided in Appendix \ref{sec:gens_opt}.
Other operations, such as updating $\pi$ and recursive calls, require constant time and use $O(n)$ space for the hash table $\pi$, tuple $T$, and sets $P$ and $C$. Therefore, {\sc GenS} generates each spanning tree of $K_n$ in constant amortized time per tree, using $O(n^2)$ space.
\end{proof}
Details of the optimizations for the algorithms {\sc GenS} and {\sc GenK} are provided in Appendix \ref{sec:gens_opt}. 
The spanning trees of $K_5$ generated by {\sc GenS} and {\sc GenK} are listed in Appendix \ref{sec:gens_k5}.

\subsection{Enumeration.}
The number of spanning trees in a complete graph $K_n$ is given by Cayley's formula, which states that there are $n^{n-2}$ spanning trees for $K_n$. 
As previously mentioned, Cayley provided an informal proof of this formula for $K_n$~\cite{Cayley1889}. Later, Pr\"{u}fer~\cite{Prufer1918} and Kirchhoff~\cite{Kirchhoff1847} offered simpler proofs by establishing a one-to-one correspondence between spanning trees and Pr\"{u}fer sequences and by applying the matrix-tree theorem, respectively. M\"{u}tze~\cite{mutze2023combinatorial}, however, commented that in \emph{The Art of Computer Programming, Volume 4}~\cite{knuth2011} Knuth sought a more detailed combinatorial explanation of Cayley's formula. 
This subsection thus provides a simpler proof of Cayley's formula for the number of spanning trees in complete graphs, based on observations from our recursive algorithm.

\begin{theorem}
    The number of spanning trees in a complete graph $K_n$ is $n^{n-2}$. 
\end{theorem}

\begin{proof}

In Algorithm~\ref{alg:spanning_complete}, {\sc GenS} takes $V$ as an input, which is a set that contains vertices not in levels 0 to $\ell-2$ of the rooted tree. 
It computes sets $P$ and $C$, which contain the vertices at level $\ell-1$ and the vertices not at levels 0 to $\ell-1$, respectively. We now derive the total number of spanning trees that can be formed given a subtree spanning levels 0 to $\ell-2$ and the set of remaining vertices $V$.

Let $k = |P|$ and $n = |C|$. 
Define $f(k, n)$ as the number of spanning trees that can be formed given a subtree spanning levels 0 to $\ell-2$, with $P$ as the set of vertices at level $\ell-1$ and $C$ as the set of vertices not in levels 0 to $\ell-1$. 
If $n = 0$, there is only one possible spanning tree, where all vertices in $V$ are at level $\ell$. 
If $n > 0$, we select $i$ vertices from $C$ to connect to vertices in $P$, where each selected vertex can connect to any of the $k$ vertices in $P$. By the product rule, there are $\binom{n}{i} k^i$ possible connections between vertices at level $\ell-1$ and vertices at level $\ell$. 
For the $n-i$ vertices in $C$ not selected to connect to $P$, there are $f(i, n-i)$ ways to form a spanning tree with the $i$ vertices at level $\ell$. 
Summing over all possible values of $i$ from 0 to $n$, the total number of spanning trees is given by the recursive formula: 
\begin{equation}\nonumber
    f(k, n) = 
    \begin{cases}
        1, & \mathrm{if~} n=0; \\
        \sum^n_{i=1} {n \choose i} k^i f(i, n-i), & \mathrm{otherwise.}
    \end{cases}
\end{equation}
We now prove by strong induction on $n$ that
\begin{equation}\nonumber
    f(k, n) = k(k+n)^{n-1}.
\end{equation}
For the base case, when $n = 0$:
\begin{align}\nonumber
    f(k, 0) &= 1 \\\nonumber
    &= k(k+0)^{-1}
\end{align}
For $n = 1$:
\begin{align}\nonumber
    f(k, 1) &= \sum^1_{i=1} {1 \choose i} k^i f(i, 1-i) \\\nonumber
    &= {1 \choose 1} k^1 f(1, 0) \\\nonumber
    &= k \cdot 1 \\\nonumber
    &= k(k+1)^{0}.
\end{align}
Assume by strong induction that $f(k, t) = k (k + t)^{t-1}$ for all $t < n$. For $t = n$:
\begin{align}\nonumber
    f(k, n) &= \sum^n_{i=1} {n \choose i} k^i f(i, n-i) \\\nonumber
    &= \sum^n_{i=1} {n \choose i} k^i \cdot in^{n-i-1} 
\end{align}
Extract $\frac{n}{i}$ from $\binom{n}{i} = \frac{n}{i} \binom{n-1}{i-1}$ for $i \geq 1$:
\begin{align}\nonumber
    f(k, n) &= \sum^n_{i=1} {n-1 \choose i-1} \cdot \frac{n}{i} \cdot k^i \cdot in^{n-i-1} \\\nonumber
    &= \sum^n_{i=1} {n-1 \choose i-1} k^i \cdot n^{n-i} \\\nonumber
    &= k\sum^n_{i=1} {n-1 \choose i-1} k^{i-1} \cdot n^{n-i}
\end{align}
By setting $j = i - 1$:
\begin{align}\nonumber
    f(k, n) &= k\sum^{n-1}_{j=0} {n-1 \choose j} k^{j} \cdot n^{n-j - 1},
\end{align}
The sum $\sum_{j=0}^{n-1} \binom{n-1}{j} k^j n^{n-1-j}$ is the binomial expansion of $(k + n)^{n-1}$, so:
\begin{equation}\nonumber
    f(k, n) = k(k+n)^{n-1}.
\end{equation}
For a complete graph $K_{n+1}$ with $n+1$ vertices, set $k = 1$ (root at level 0) and $|C| = n$ (remaining vertices), so:
\begin{equation}\nonumber
f(1, n) = 1 \cdot (1 + n)^{n-1} = (n + 1)^{n-1}.
\end{equation}
Adjusting for the complete graph $K_n$, we correct the index: $f(1, n-1) = 1 \cdot (1 + (n-1))^{n-2} = n^{n-2}$, which matches Cayley's formula for $K_n$.
\end{proof}

\section{Edge-exchange Gray codes for spanning tree of general graphs.}
\label{sec:algorithm}
With modifications, our algorithm {\sc GenS} can be extended to generate edge-exchange Gray codes for connected general graphs. 
The approach for generating edge-exchange Gray codes for general graphs is similar to that for generating pivot Gray codes for complete graphs. 
However, a key difference arises between complete and general graphs. 
In {\sc GenS}, given sets $P$ and $C$ in the rooted tree, where $P$ contains vertices at level $\ell-1$ and $C$ contains vertices not in levels 0 to $\ell-1$, every vertex in $C$ can connect to any vertex in $P$ in a complete graph, yielding $|P|^{|C|}$ possible configurations for connections between levels $\ell-1$ and $\ell$. 
In general graphs, however, not all such connections may exist, as an edge between $p \in P$ and $c \in C$ may be absent in the original graph. 
Thus, additional handling is required to account for valid connections between vertices in $P$ and $C$.

Below, we describe three mechanisms to modify our algorithm {\sc GenS} to generate edge-exchange Gray codes for connected general graphs. We illustrate the handling required for each case using the example graph $G$ shown in Figure~\ref{fig:general_graph}.
\begin{figure}[H]
    \centering
    \includegraphics[width=0.55\linewidth]{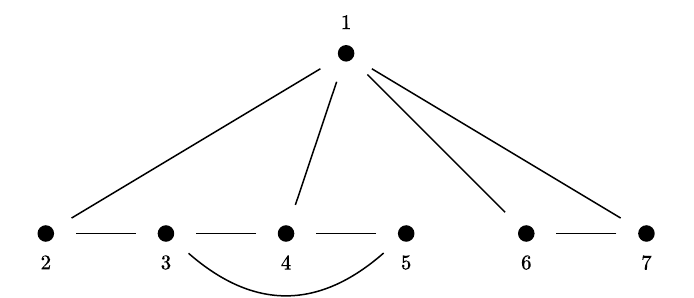}
    \caption{A connected general graph $G$ with seven vertices.}
    \label{fig:general_graph}
\end{figure}

The three modifications to {\sc GenS} are discussed in the following three subsections.

\subsection{Mixed-radix string representation.}

As previously noted, given sets $P$ and $C$ in the rooted tree, where $P$ contains vertices at level $\ell-1$ and $C$ contains vertices not in levels 0 to $\ell-1$, not all connections between vertices in $P$ and $C$ may exist, as an edge between $p \in P$ and $c \in C$ may be absent in the original graph. The challenge is to represent all valid configurations when some connections are not possible.

Consider the spanning tree of the graph $G$ in Figure~\ref{fig:general_graph} and its corresponding rooted tree shown in Figure~\ref{fig:enter-label}.  
\begin{figure}[H]
    \centering
    \includegraphics[width=0.3\linewidth]{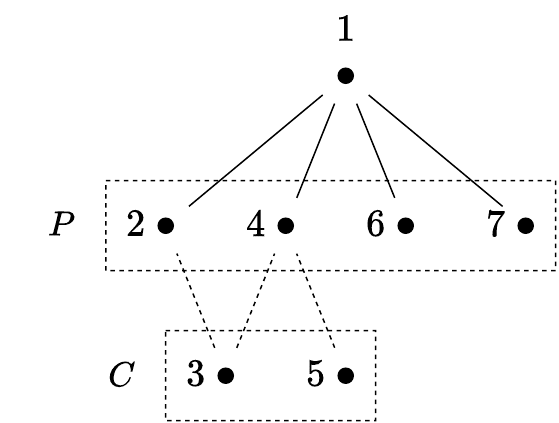}
    \caption{The rooted tree that corresponds to a spanning tree of $G$. }
    \label{fig:enter-label}
\end{figure} 
\noindent At levels 1 and 2 of the rooted tree, we have \( P = \{2, 4, 6, 7\} \) and \( C = \{3, 5\} \). 
In this case, vertex 3 has two possible connections to vertices in \( P \), while vertex 5 has only one possible connection.
Unlike in complete graphs, where \( t_\ell \) is a \( |C| \)-length string with each character \( a_i \) ranging from 0 to \( |P| \), here each \( a_i \) ranges from 0 to the number of valid connections for the corresponding vertex in \( C \). 
We thus use mixed-radix strings to encode parent assignments between vertices in $P$ and $C$ for general graphs, unlike using $k$-ary strings for complete graphs.

A \emph{mixed-radix string} is a sequence where each position has its own base (or radix), determining the range of valid values for that position. 
Each element \( t_\ell = a_1 a_2 \cdots a_{n-q_{\ell-1}} \) in the tuple \( T \) is now a mixed-radix string of length \( n-q_{\ell-1} \), where each \( a_i \) ranges from 0 to the number of vertices in \( P \) adjacent to the \( i \)-th smallest vertex in \( C \). 
If \( a_i = 0 \), then the \( i \)-th smallest vertex at level \( \ell \) is not connected to any vertex at level \( \ell-1 \). 
Otherwise if \( a_i > 0 \), then it is connected to the \( (a_i) \)-th smallest vertex in \( P \) that is adjacent to the \( i \)-th smallest vertex in \( C \) in the original graph \( G \).
For example, in Figure~\ref{fig:general_graph}, if $P=\{2, 4, 6, 7\}$ and $C=\{3, 5\}$ where vertex 3 is adjacent to vertices 2 and 4 in \( G \), and vertex 5 is adjacent to vertex 4 in $G$, then \( t_\ell \) is a string of length two with \( a_1 \in \{0, 1, 2\} \) (representing no connection or connections to 2 or 4) and \( a_2 \in \{0, 1\} \) (representing no connection or a connection to 4). 
The possible configurations for vertices 3 and 5 can be represented by a mixed-radix string of length two, with the first digit in base three (for vertex 3’s connections) and the second digit in base two (for vertex 5’s connection). 
The valid mixed-radix strings representing these configurations are \( 01, 10, 11, 20, 21 \). 
Note that the string \( 00 \) is omitted from the listing, as a spanning tree requires at least one connection between levels \( \ell-1 \) and \( \ell \).

\begin{algorithm}[t]
    \caption{Generating mixed-radix strings starting with $\alpha$ in reflectable Gray code order. }
    \label{alg:kg-ary}
    \begin{algorithmic}[1]
        \Procedure{GenMR}{$i, k=k_1k_2 \cdots k_j, \alpha=a_1a_2 \cdots a_j$}
        \State{$i^\prime \gets m[i]$}
            \For{$r = 1$ to $k_{i^\prime}$}
                \If{$i = \vert \alpha \vert$}
                    \If{$s \neq 0$}
                    \State{Output $\alpha, \delta$}
                    \State{$\delta \gets \emptyset$}
                    \EndIf
                \Else
                    \State \Call{GenMR}{$i + 1, k, \alpha$}
                    \State{$j^\prime \gets m[i+1]$}
                    \If{$\alpha_{j^\prime} = k_{j^\prime} - 1$}
                        {$d_{j^\prime} \gets 1$}
                    \Else
                        {$~d_{j^\prime} \gets -d_{j^\prime}$}
                    \EndIf
                \EndIf
                \If{$r < k_{i^\prime}$}
                    \State $s \gets s + \big((\alpha_{i^\prime} + d_{i^\prime} + k_{i^\prime}) \mod k_{i^\prime}\big) - \alpha_{i^\prime}$
                    \State $\alpha_{i^\prime} \gets (\alpha_{i^\prime} + d_{i^\prime} + k_{i^\prime}) \mod k_{i^\prime}$
                    \State $\delta \gets \delta, i$ 
                \EndIf
            \EndFor
        \EndProcedure
    \end{algorithmic}
\end{algorithm}

To generate a Gray code for spanning trees of general graphs, we generate a Gray code for mixed-radix strings representing connections between vertices in \( P \) and \( C \), satisfying the three conditions in Section~\ref{sec2:pivotGenS}.
\begin{lemma}
The set of mixed-radix strings of length $j$ is a reflectable language.
\end{lemma}\label{lem:reflectableG}
\begin{proof}
Let \( k_i \) denote the radix of the character \( a_i \) in a mixed-radix string \( a_1 a_2 \cdots a_j \) of length \( j \), where \( 0 \leq a_i < k_i \) for each \( i \). 
 By setting $x_i = k_i - 1$ and $y_i = k_i - 2$ for each $i$, observe that mixed-radix strings of length $j$ are reflectable.
\end{proof}
The algorithm {\sc GenK} is modified to create {\sc GenMR}, which generates a Gray code for mixed-radix strings by replacing the parameter $k$ with an array $(k_1, k_2, \ldots, k_j)$, where $k_i \geq 1$ specifies the range of each digit $a_i$ in a string $\alpha = a_1 a_2 \cdots a_j$, with $a_i \in \{0, 1, \ldots, k_i - 1\}$. 
To ensure spanning tree connectivity, the all-zero string $0^j$ is excluded, requiring at least one non-zero digit. 
To omit $0^j$, we apply a simple trick when there exists an index $i < j$ with $k_i \geq 2$ and the last position has $k_j = 2$. 
In this case, we swap the levels of $i$ and $j$ when we generate the reflectable Gray code for mixed-radix strings.
This swap ensures that the string \( 0^j \) appears between the strings \( 0^{j-1}x \) and \( 0^{j-1}y \), where \( x \neq y \), allowing us to remove \( 0^t \) while maintaining a 1-Gray code. 
A graphical interpretation of this simple trick is illustrated in Figure~\ref{fig:MRtrick}.
It is also worth noting that if there exists $k_i = 1$ for some position $i$, then $a_i = 0$ is the only possibility, which does not affect the 1-Gray code property of the resulting listing.  
Pseudocode of the algorithm is shown in Algorithm~\ref{alg:kg-ary}.  
To run the algorithm, we make the initial call {\sc GenMR}$(i, k , \alpha)$, where $i$ is the current level of the recursive computation tree (initialized to 1), $k = (k_1, k_2, \ldots, k_j)$ is the array of radices for each digit in $\alpha$, and $\alpha$ is the initial string of the listing. 
A complete Python implementation of {\sc GenMR} is provided in Appendix \ref{sec:genk_code}.

\begin{figure}[t]
    \centering
    \includegraphics[width=0.8\linewidth]{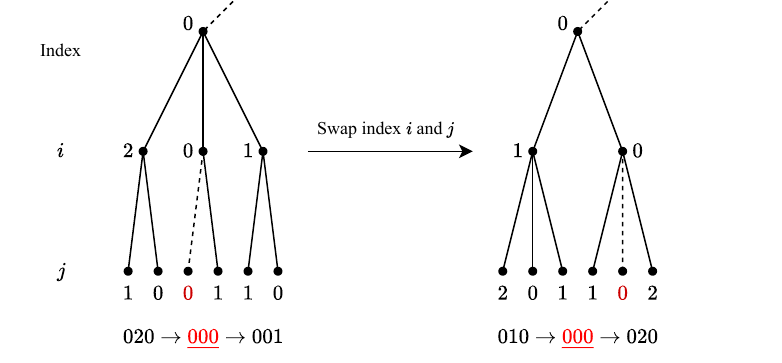}
    \caption{Swapping the ordering of $c_i$ and $c_j$ in $C$ to generate a reflectable Gray code. }
    \label{fig:MRtrick}
\end{figure}

\begin{lemma} \label{lem:genk_gcg}
Consecutive strings generated by \textsc{GenMR} differ by a single symbol change, or by swapping a 0 and a 1 which the strings have only one single non-zero character. 
\end{lemma}

\begin{proof}
The proof follows a similar approach to that of Lemma~\ref{lem:genk_gc}. 
If $k_i = 2$ for all $i$, the listing is binary, and \textsc{GenMR} generates the same listing as \textsc{GenK} in Algorithm~\ref{alg:k-ary}. 
Thus, consecutive strings differ by a single bit change or a swap between a 0 and a 1. 
If there exists some $k_i = 1$, then $a_i = 0$, fixing the digit 0 at that level $i$, and the fixed digit 0 does not affect the 1-Gray code property. 
Otherwise, if there exists some $k_i > 2$, our simple trick on swapping $c_i$ with $c_j$ in $C$ with $k_j= 2$ ensures that the last character has base $k_j > 2$. 
By Lemma~\ref{lem:reflectableG}, matrix-radix strings form a reflectable language, and \textsc{GenMR} produces a 1-Gray code when the string $0^j$ is included in the listing. 
Observe that $0^j$ appears between strings $0^{j-1}x$ and $0^{j-1}y$ when $k_j > 2$, where $x \neq y$ and $x, y \neq 0$. Removing the string $0^j$ maintains the 1-Gray code property, as the transition from $0^{j-1}x$ to $0^{j-1}y$ involves a single symbol change in the last position.
\end{proof}

\subsection{Connected components.}

The second modification addresses the connectivity of subtrees. 
Consider the initial spanning tree of the graph $G$ in Figure~\ref{fig:general_graph}. 
In our algorithm {\sc GenS} for complete graphs, the initial spanning tree is the path $1-2-\cdots-n$, corresponding to the tuple $T = (10^{n-1}, 10^{n-2}, \ldots, 1)$. 
However, such an initial tree is not feasible for $G$. 
For instance, if vertex 1 is only connected to vertex 2, vertices 6 and 7 become disconnected from the rest of $G$. 
Thus, we design a mechanism to ensure that the spanning tree accounts for the connectivity of each component in $C$ with respect to the vertices in $V$.

Given a vertex set \( C \) at level \( \ell \), we apply depth-first search or breadth-first search to identify the connected components \( \{C_1, C_2, \ldots, C_m\} \) within \( C \). 
Each connected component \( C_i \) must have at least one vertex connected to a parent in \( P \); otherwise, it would be disconnected from the rest of the graph, preventing the formation of a spanning tree. 
We denote \( t_{\ell, i} \) as a mixed-radix string of length \( |C_i| \), where each character represents the connection of the corresponding vertex in \( C_i \) to a vertex at level \( \ell-1 \). 
To ensure connectivity, the string \( 0^{|C_i|} \), which indicates no connections to \( P \), is excluded for each \( t_{\ell, i} \). 
We then apply the algorithm {\sc GenMR} to generate all possible mixed-radix strings \( t_{\ell, i} \) for each connected component \( C_i \), excluding \( 0^{|C_i|} \). 
Thus, {\sc GenMR} enables the generation of all possible connection configurations for the rooted tree, ensuring that each component \( C_i \) maintains at least one connection to \( P \).
As an example, consider the initial spanning tree for the graph \( G \). 
At level zero, \( P = \{1\} \), and we consider level one. 
The algorithm performs a depth-first search or a breadth-first search to identify the connected components in \( C \), yielding \( C_1 = \{2, 3, 4, 5\} \) and \( C_2 = \{6, 7\} \). 
For \( C_1 \), only vertices 2 and 4 are adjacent to vertex 1 in \( P \), so \( t_{1, 1} \) is a mixed-radix string of length four, with characters for vertices 2 and 4 in \( \{0, 1\} \) (no connection or connection to vertex 1) and characters for vertices 3 and 5 in \( \{0\} \) (no connection possible). 
For \( C_2 \), both vertices 6 and 7 are adjacent to vertex 1, so \( t_{1, 2} \) is a mixed-radix string of length two, with each character in \( \{0, 1\} \) (no connection or connection to vertex 1). 
The valid strings for each component exclude \( 0^{|C_i|} \), ensuring at least one connection to \( P \).

\begin{figure}
    \centering
    \includegraphics[width=0.35\linewidth]{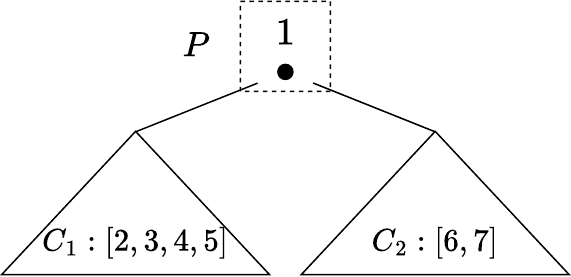}
    \caption{A rooted tree of $G$ rooted at vertex 1 with two connected components $C_1 = \{ 2 , 3 , 4 , 5 \}$ and $C_2 = \{ 6 , 7 \}$ at level 1.}
    \label{fig:componentG}
\end{figure}

The case for complete graph can actually be considered as a special case of the algorithm to generate spanning trees for general graphs, where there is only one single connected component in a complete graph. 

\subsection{Two-dimensional recursion for connected components.}
For complete graphs, the algorithm {\sc GenS} traverses the rooted tree level by level, generating pivot Gray codes for spanning trees. For general graphs, we account for connected components, introducing a second dimension to the recursion.

We design a two-dimensional recursion to iterate through all possible mixed-radix strings $t_{\ell, i}$ for each connected component $C_i$ at level $\ell$ and all configurations at subsequent levels. In the first dimension, the algorithm iterates over all possible mixed-radix strings $t_{\ell, i}$ for each connected component $C_i$ at level $\ell$ by invoking {\sc GenMR} to generate each $t_{\ell, i}$. For each update to $t_{\ell, i}$, we recursively call {\sc GenMR} to generate all possible mixed-radix strings $t_{\ell, i+1}$ for the next connected component $C_{i+1}$. In the second dimension, after each update to any $t_{\ell, i}$ at level $\ell$, the algorithm descends to level $\ell+1$, identifies all connected components at level $\ell+1$ based on the current configuration at level $\ell$, and recursively generates all possible mixed-radix strings $t_{\ell+1, i}$ for each connected component at level $\ell+1$ and possible configurations at subsequent levels.

We now demonstrate how the algorithm generates spanning trees for the graph $G$. Initially, vertex 1 is at level 0 of the rooted tree, and the algorithm enumerates all possible strings $t_{\ell, i}$ for each connected component at level $\ell$. For the connections between levels 0 and 1, we have $P = \{1\}$ and $C = \{2, 3, 4, 5, 6, 7\}$, with two connected components: $C_1 = \{2, 3, 4, 5\}$ and $C_2 = \{6, 7\}$. 
Since only vertices 3 and 5 in $C_1$ are adjacent to vertex 1 in $G$, the string $t_{1,1}$ for $C_1$ is a mixed-radix string of length four with radices $k_1 = 1$ (for vertex 2), $k_2 = 2$ (vertex 3), $k_3 = 1$ (vertex 4), and $k_4 = 2$ (vertex 5). 
Since both vertices 6 and 7 in $C_2$ are adjacent to vertex 1, the string $t_{1,2}$ for $C_2$ is a binary string of length two with radix $k_i = 2$ for each position. In the first dimension of the recursion, the algorithm generates all combinations of $t_{\ell, i}$, starting with $t_{1,1}$ for $C_1$. For each update to $t_{1,1}$, it generates all strings in $t_{1,2}$ for $C_2$ and all configurations at subsequent levels. 
Initially, $t_{1,1} = 0100$ (vertex 3 in $C_1$ connects to vertex 1) and $t_{1,2} = 10$ (vertex 6 in $C_2$ connects to vertex 1), yielding trees 1, 2, and 3 in Figure~\ref{fig:spanningG}. After exhausting all configurations for subsequent levels, the algorithm updates $t_{1,2}$ to $11$ (both vertices 6 and 7 connect to vertex 1), producing trees 4, 5, and 6 in Figure~\ref{fig:spanningG}. Next, it updates $t_{1,2}$ to $01$ (vertex 7 connects to vertex 1), yielding trees 7, 8, and 9 in Figure~\ref{fig:spanningG}. After exhausting all strings in $t_{1,2}$, the algorithm updates $t_{1,1}$ to $0101$ (vertices 3 and 5 connect to vertex 1) and sets $t_{1,2} = 01$, generating trees 10 to 14 in Figure~\ref{fig:spanningG}. It then updates $t_{1,2}$ to $10$ (vertex 6 connects to vertex 1) and continues this process until all spanning trees of $G$ are generated.

\subsection{Analysis of the edge-exchange Gray code generation algorithm for general graphs.}

The algorithm {\sc GenG} extends {\sc GenS} (Algorithm~\ref{alg:spanning_complete}) to generate Gray codes for spanning trees of general graphs by incorporating the changes described in above subsections. 
The resulting listing forms an edge-exchange Gray code for the spanning trees of a general graph $G$ but does not form a pivot Gray code. For example, in case (b) described in Section~\ref{sec2:pivotGenS}, a vertex $v$ at level $\ell$ may be the only vertex in its subtree with an edge to a vertex at level $\ell-1$, while a descendant $x \neq v$ has edges to other vertices at level $\ell$ within the same connected component (e.g., tree 24 and tree 25 in Figure~\ref{fig:spanningG}, also see Figure~\ref{fig:spanningG24}). Removing the edge $(v, u)$ to its parent $u$ at level $\ell-1$ requires adding an edge $(x, w)$, where $w \neq v$ is in the same component, to maintain connectivity. This edge exchange does not involve a common vertex, violating the pivot Gray code property, which requires consecutive trees to differ by a single edge exchange involving a common vertex.

\begin{figure}[t]
    \centering
    \includegraphics[width=1\linewidth]{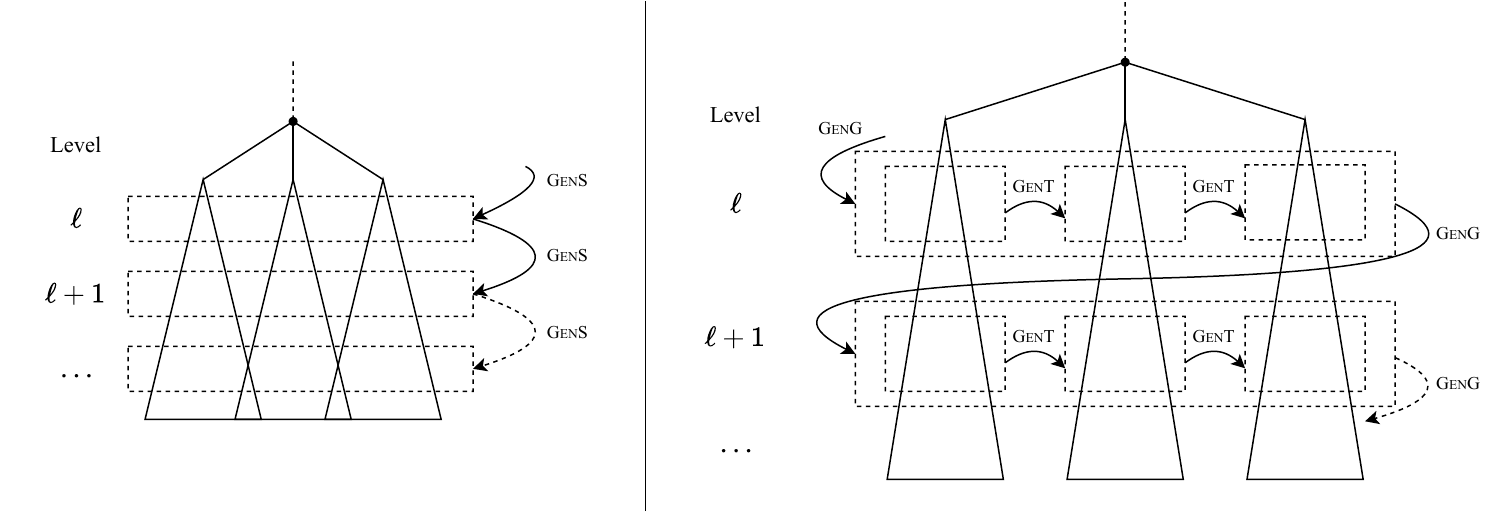}
    \caption{Two-dimensional recursion to iterate through all possible mixed-radix strings $t_{\ell, i}$ for each connected component $C_i$ at level $\ell$ and all configurations at subsequent levels.}
    \label{fig:two_dim}
\end{figure}

The following modifications are applied to {\sc GenS} to produce {\sc GenG}:

\begin{itemize}

\item At each level $\ell$ of the rooted tree, apply depth-first search or breadth-first search to identify the connected components in $C$ after computing the set $C$ (line 7 of {\sc GenG}, Algorithm~\ref{alg:spanning_general});

\item For each vertex $c_{ij}$ in connected component $C_i$, maintain a list $F_{ij}$ of vertices in $P$ to which $c_{ij}$ is adjacent in the graph $G$ (lines 10--12 of {\sc GenG});

\item Introduce a function {\sc GenT} (Algorithm~\ref{alg:spanning_general2}) to recursively generate all mixed-radix strings $t_{\ell, i}$ for each connected component $C_i$, mirroring lines 11--20 of {\sc GenS}. The function {\sc GenT} takes two parameters, where $i$ is the index of the current component $C_i$ in $\mathcal{C}$, and $\ell$ is the current level of the rooted tree;

\item Initiate the recursion by calling {\sc GenT}($1,  \ell$) to generate mixed-radix strings for all connected components at level $\ell$ (line 13 of {\sc GenG});

\item Within {\sc GenT}, recursively call {\sc GenT}($i+1, \ell$) to generate all mixed-radix strings for the next connected component $C_{i+1}$ (line 7 and line 21 of {\sc GenT});

\item For each update to the mixed-radix string $t_{\ell, i}$ in $C_i$, call {\sc GenG} to generate configurations for the next level $\ell+1$ when $i > |\mathcal{C}_\ell|$ where $\mathcal{C}_\ell$ is the number of connected components in level $\ell$, using the union of components $\bigcup C_i$ and the concatenated string $\prod t_i$ (lines 2--5 of {\sc GenT});

\item Determine the base (number of possible parents) for each vertex $c_{ij}$ in $C_i$, derived from $|F_{ij}| + 1$ to account for connections within $C_i$ (line 8 of {\sc GenT});

\item Use mixed-radix strings $t_{\ell, i}$ for each connected component $C_i$ instead of $k$-ary strings, with the {\sc NextMR} function invoking {\sc GenMR} to generate the next mixed-radix string (lines 10--11 of {\sc GenT});

\item For case (b), where $t_{i,\delta_0} = 0$, select an edge $(c_{a}, c_{b}) \in G$ such that $c_{a} = c_{\delta_0}$, or $c_{a}$ is a descendant of $c_{\delta_0}$ while $c_{b}$ is not a descendant of $c_{\delta_0}$, setting $c_{b}$ as the parent of $c_{a}$ (lines 18--20 of {\sc GenT}). 
\end{itemize}
Pseudocode of {\sc GenG} and {\sc GenT} to generate all spanning trees of general graphs is shown in Algorithm~\ref{alg:genG} and Algorithm~\ref{alg:genT}.
A complete Python implementation is provided in Appendix \ref{sec:geng_code}. 
\begin{figure}
\centering
\begin{minipage}{\textwidth}
\begin{algorithm}[H]
    \caption{
    Generating an edge-exchange Gray code for spanning trees of a general graph}
    \label{alg:spanning_general}
    \begin{algorithmic}[1]
        \Procedure{\sc GenG}{$V_{\ell-1}, t_{\ell-1}=a_1a_2\cdots a_{\vert V_{\ell-1} \vert}$}
        \State{$P_{\ell} \gets v_i \in V_{\ell-1} \mathrm{~where~} a_i > 0$}
        \State{$C_{\ell} \gets v_i \in V_{\ell-1} \mathrm{~where~} a_i = 0$}
        \If{$\vert C_{\ell} \vert = 0$}
            \State{Print spanning tree}
            \State{\textbf{return}}
        \EndIf
        \State{$\mathcal{C}_{\ell} \gets $ Find connected components in $C_{\ell}$}
        \For{$C_{\ell, i} \in \mathcal{C}_{\ell} = \big\{{C}_{\ell, 1},\dots C_{\ell, m}\big\}$}
        \State{$t_{\ell,i} = b_1b_2\cdots b_{\vert C_{\ell, i} \vert}\gets 0^{\vert C_{\ell, i} \vert}$}
        \For{$c_{ij} \in C_{\ell, i}$}
        \State{$F_{\ell, ij} \gets p_z \in P_{\ell} \mathrm{~where~} (p_z, c_{ij}) \in G$}
        \If{$\pi(c_{ij}) \in F_{\ell, ij}$}{~$t_{\ell, ij} \gets z \mathrm{~where~} \pi(c_{ij}) = F_{\ell, ij}[z]$}
        \EndIf
        \EndFor
        \EndFor

        \State{{\sc GenT}($1, \ell$)}
    
        \EndProcedure
    \end{algorithmic}\label{alg:genG}
\end{algorithm}
\begin{algorithm}[H]
    \caption{Generate all mixed-radix strings $t_{\ell, i}$ for each connected component $C_i$.}
    \label{alg:spanning_general2}
    \begin{algorithmic}[1]
        \Procedure{\sc GenT}{$i, \ell$}

        \If{$i > \vert \mathcal{C}_{\ell} \vert$}
        \State{$V_\ell \gets \bigcup_{i \le \vert \mathcal{C}_{\ell} \vert} C_{\ell, i}$}
        \State{$t_{\ell} \gets \prod_m t_{\ell, i}$}
        \State{{\sc GenG}$(V_{\ell}, t_{\ell}$)}
        \State{\textbf{return}}
        \EndIf
        \State{{\sc GenT}$(i+1, \ell)$}
        \State{$k_{\ell, i} \gets \big[ \vert F_{\ell, ij}\vert +1  \big]_{j \in [1, \dots, \vert F_{\ell, i}\vert ]}$}
        \State{$\big\{c_{1}, c_{2}, \dots\big\} \gets C_{\ell, i}$}
        
        \While{{\sc NextMR}$(t_{\ell, i}, k_{\ell, i}) \ne \emptyset$}
            \State{$t_{\ell, i}, \delta \gets ${\sc NextMR}$(t_{\ell, i}, k_{\ell, i})$}
            \If{$\vert \delta \vert > 1$}
                \State{$\pi(c_{\delta_1}) \gets F_{i\delta_1}[t_{\ell, i\delta_1}]$}{\color{blue}$\quad \quad \quad  \quad  \quad  \quad  \quad \quad \quad \quad \quad  \quad  \quad  \quad  \quad \quad \quad \quad \quad  \quad  \quad  \quad  \quad \quad \hspace{1em} \triangleright  \ $case (d)}
                \State{{\sc Lift}$(c_{\delta_1})$}
            \ElsIf{$k_{\delta_0} > 0$}
                \State{$\pi(c_{\delta_0}) \gets F_{i\delta_0}[t_{i\delta_0}]$}{\color{blue}$\quad \quad \quad  \quad  \quad  \quad  \quad \quad \quad \quad \quad  \quad  \quad  \quad  \quad \quad \quad \quad \quad  \quad  \quad  \quad  \quad \quad \quad \hspace{0.5em} \triangleright  \ $case (a) and (c)}
            \Else
                \State{Find an edge $(c_{a}, c_{b}) \in G$ where $\exists \beta \ge 0, c_{\delta_0} = \pi^\beta(c_{a})$ and $\nexists \gamma , c_{\delta_0} = \pi^\gamma(c_{b}) $}{\color{blue}$ \quad \hspace{0.2em} \triangleright  \ $case (b)}                \State{$\pi(c_{a}) \gets \pi(c_{b})$}
                \State{{\sc Lift}$(c_{a})$}
            \EndIf
            \State{{\sc GenT}$(i+1, \ell)$}
        \EndWhile
        \EndProcedure
    \end{algorithmic}\label{alg:genT}
\end{algorithm}
\end{minipage}
\end{figure}

\begin{figure}[t]
    \centering
    \includegraphics[width=0.65\linewidth]{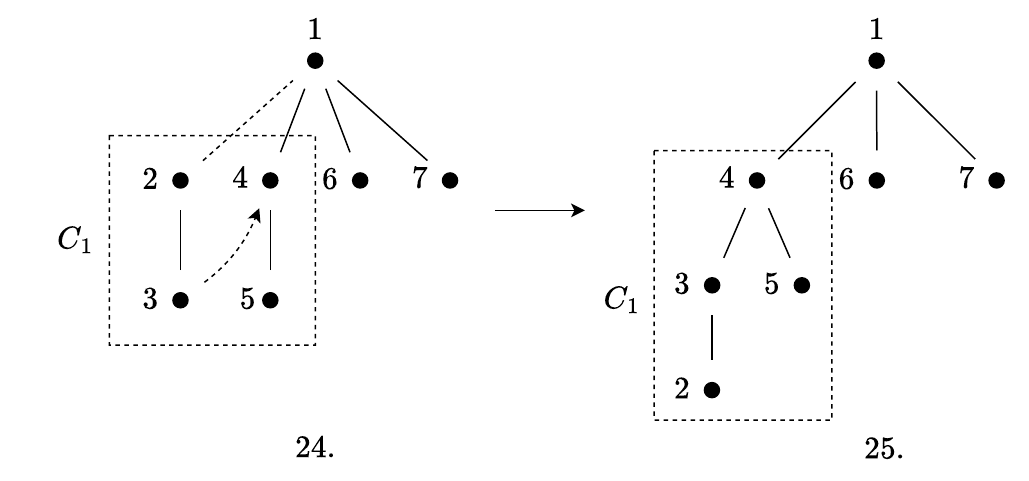}
    \caption{Tree 24 and tree 25 generated by {\sc GenG} for $G$. The change between tree 24 and tree 25 is an edge-exchange but not a pivot change.}
    \label{fig:spanningG24}
\end{figure}

\begin{theorem}
The algorithm {\sc GenG} generates an edge-exchange Gray code for the spanning trees of a general graph $G$.
\end{theorem}

\begin{proof}
The algorithm {\sc GenG}, with {\sc GenT}, applies the four cases in Lemma~\ref{lem:diff} to generate spanning trees of $G$. We analyze the edge exchanges in each case to confirm they produce an edge-exchange Gray code, where consecutive trees differ by exactly one edge exchange.

\begin{itemize}
\item {Case (a)}: $a_v$ updates from $0$ to $r > 0$. A vertex $v$ at level $q > \ell$ is disconnected from its parent $w$ and connected to a vertex $u$ at level $\ell-1$ where $u$ is the $r$-th smallest available parent for $v$. This removes edge $(v, w)$ and adds $(v, u)$, a pivot change (and thus an edge-exchange change), consistent with complete graphs in {\sc GenS};
\item {Case (c)}: $a_v$ updates from $r > 0$ to some $j > 0, j \ne r$. A vertex $v$ at level $\ell$ is disconnected from its parent $u$ at level $\ell-1$ and connected to another vertex $w$ at level $\ell-1$. This removes $(v, u)$ and adds $(v, w)$, a pivot change (and thus an edge-exchange change), consistent with complete graphs in {\sc GenS};
\item {Case (d)}: A bit swap updates $a_v = 1, a_w = 0$ to $a_w = 1, a_v = 0$. A vertex $u_1$ at level $\ell-1$ disconnects from its child $v$ at level $\ell$, reflecting $a_v \gets 0$. A vertex $w$ at level $q > \ell$, which is a descendant of $v$, disconnects from $\pi(w)$ and connects to the first valid parent $u_2$ at level $\ell$, reflecting $a_w \gets 1$. This removes $(v, u_1)$ and adds $(w, u_2)$ where $u_1$ may not be $u_2$, satisfying the edge-exchange Gray code property but violating the pivot Gray code property. The bit swap generated by {\sc GenMR} involves only swapping a 0 and a 1, and strings have only one non-zero character (Lemma~\ref{lem:genk_gcg}), the subtree rooted at $v$ must remain connected after the bit swap;

\item {Case (b)}: $a_v$ updates from $r > 0$ to $0$. A vertex $v$ at level $\ell$ is disconnected from its parent $u$ at level $\ell-1$, and the subtree rooted at $v$ is reconnected to a vertex $w$ in the same connected component $C_i$ at level $q > \ell$. 
In general graphs, $v$ may not be able to connect to such a vertex $w$ which is at level $q \ge \ell$ and not in the subtree rooted at $v$. 
Removing the edge $(v, u)$ requires adding an edge $(x, w)$, where $x$ is a descendant of $v$, to maintain connectivity.
Such a connection $(x, w)$ must exist because there are more than one vertex in the the same connected component $C_i$ at level $\ell$, as $t_{\ell,i}$ contains more than one non-zero character (otherwise, it would result in the string $0^{|t_{\ell,i}|}$, which {\sc GenMR} avoids generating).  
This edge exchange does not involve a common vertex when $v \neq x$, satisfying the edge-exchange Gray code property but violating the pivot Gray code property.

\end{itemize}

Thus, {\sc GenG} ensures that consecutive spanning trees differ by a single edge exchange, satisfying the edge-exchange Gray code property. To establish that {\sc GenG} generates all spanning trees of $G$ exactly once, we proceed by contradiction, following a similar approach to Theorem~\ref{thm:GrayKN} for complete graphs.

Suppose there exists a spanning tree $T$ of $G$ that does not appear in the listing generated by {\sc GenG}. Such a spanning tree can be encoded as a tuple $T = ((t_{1,1}, t_{1,2}, \ldots, t_{1,|C_1|}), (t_{2,1}, \ldots, t_{2,|C_2|}), \ldots, (t_{h,1}, t_{h,2}, \ldots, t_{h,|C_h|}))$, where $|C_i|$ denotes the number of connected components at level $i$ with respect to the subtree $S_{i-1}$ spanning levels 0 to $i-1$. If $T$ is absent, no spanning tree in the listing has a subtree $S_{h-1}$, spanning levels 0 to $h-1$, matching the first $h-1$ elements $((t_{1,1}, t_{1,2}, \ldots, t_{1,|C_1|}), \ldots, (t_{h-1,1}, t_{h-1,2}, \ldots, t_{h-1,|C_{h-1}|}))$ of $T$. Otherwise, {\sc GenT} would exhaustively generate a Gray code for $(t_{h,1}, t_{h,2}, \ldots, t_{h,|C_h|})$ to include $T$. Recursively applying this reasoning, no spanning tree has a subtree $S_{h-2}$ matching the first $h-2$ elements $((t_{1,1}, t_{1,2}, \ldots, t_{1,|C_1|}), \ldots, (t_{h-2,1}, t_{h-2,2}, \ldots, t_{h-2,|C_{h-2}|}))$ of $T$. 
Continuing applying the same reasoning implies no spanning tree matches $((t_{1,1}, t_{1,2}, \ldots, t_{1,|C_1|}))$. However, at level zero, the depth-first search or breadth-first search in {\sc GenG} (line 7) identifies all connected components, and {\sc GenT} generates all valid mixed-radix strings $t_{1,i}$ (via {\sc GenMR}), excluding $0^{|C_i|}$, covering all possible configurations for $(t_{1,1}, t_{1,2}, \ldots, t_{1,|C_1|})$. This contradicts the assumption that $T$ is missing. Thus, all spanning trees are generated exactly once, and the listing is an edge-exchange Gray code.
\end{proof}

\begin{theorem}
Spanning trees of complete graphs $G$ with $n$ vertices can be generated in edge-exchange Gray code order in $O(n^2)$ per tree, using $O(n^2)$ space.
\end{theorem}
\begin{proof}
The algorithm {\sc GenG}, with {\sc GenT}, extends {\sc GenS} (Algorithm~\ref{alg:spanning_complete}), which generates spanning trees of a complete graph $K_n$ in constant amortized time per tree using $O(n^2)$ space. For general graphs, {\sc GenG} introduces the following additional steps that increase the time complexity while maintaining the space complexity:
\begin{figure}[]
    \centering
    \includegraphics[width=0.84\linewidth]{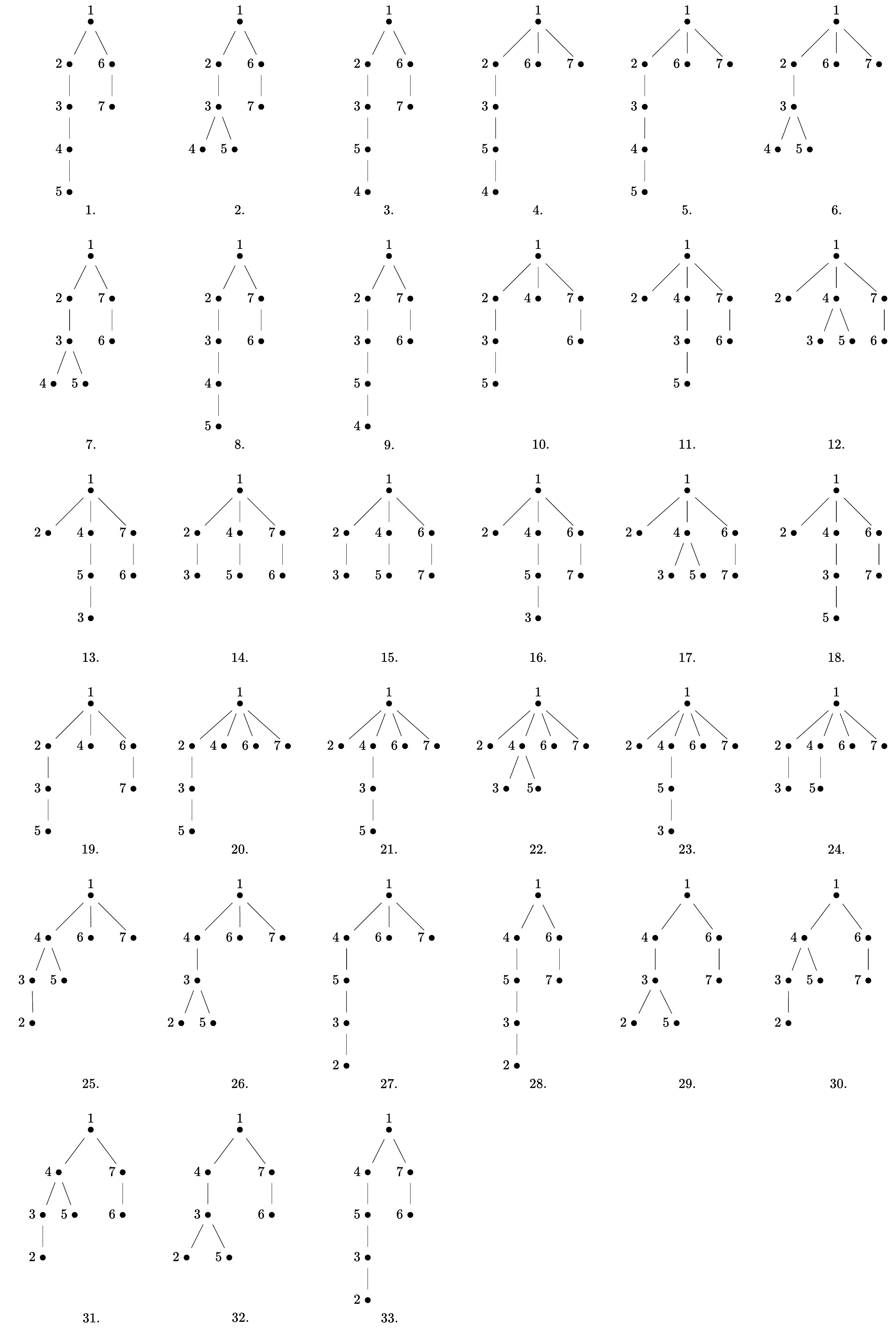}
    \caption{An edge-exchange Gray code for spanning trees of the graph $G$ in Figure \ref{fig:general_graph}.}
    \label{fig:spanningG}
\end{figure}
\begin{itemize}
\item Identify connected components in $C$ using depth-first search or breadth-first search, which takes $O(n^2)$ time in the worst case (line 7 of {\sc GenG}, Algorithm~\ref{alg:spanning_general});
\item For each vertex $c_{i,j}$ in the connected component $C_i$, construct the list $F_{i,j}$ of vertices in $P$ adjacent to $c_{i,j}$ in $G$, which takes $O(n^2)$ time across all vertices (lines 8--12 of {\sc GenG});
\item In case (b), identify a descendant $u$ of vertex $v$ and an edge $(u, w)$ to maintain connectivity, which takes $O(n^2)$ time in the worst case (lines 18--20 of {\sc GenT}, Algorithm~\ref{alg:spanning_general2}).
\end{itemize} 
In the worst case, the third step may be required for each spanning tree generated. Other operations, such as updating $t_{\ell,i}$ via {\sc GenMR} and recursive calls, take constant amortized time, consistent with {\sc GenS}. Thus, the total time per tree is dominated by the $O(n^2)$ operations, while the space usage remains $O(n^2)$ for storing $P$, $C$, $\mathcal{C}$, $t$, and $F$. Hence, {\sc GenG} generates spanning trees of a general graph $G$ in edge-exchange Gray code order in $O(n^2)$ time per tree, using $O(n^2)$ space.
\end{proof}

 The edge-exchange Gray code for spanning trees of the Petersen graph is provided in Appendix E. Each spanning tree is represented as a rooted tree with root at vertex 1 using a \emph{compact representation}, where a tree with $n$ vertices is encoded as a string $a_1 a_2 \cdots a_{n-1}$ of length $n-1$, with $a_i$ denoting the parent of vertex $i+1$. The compact representations of the rooted trees for graph $G$ in Figure \ref{fig:general_graph} are listed in Figure~\ref{fig:spanningG} as follows:

 \begin{center}
     123416, 123316, 125316, 125311, 123411, 123311, 123371, 123471, 125371, 121371, 141371,\\ 141471, 151471, 121471, 121416, 151416, 141416, 141316, 121316, 121311, 141311, 141411,\\ 151411, 121411, 341411, 341311, 351411, 351416, 341316, 341416, 341471, 341371, 351471. 
 \end{center}

\section{Final remarks.}
\label{sec:remark}

The algorithm {\sc GenG} generates spanning trees of a general graph \( G \) with \( n \) vertices in edge-exchange Gray code order in \( O(n^2) \) time per tree, using \( O(n^2) \) space. Optimizations are possible when specific graph properties are known. For instance, for complete graphs \( K_n \), {\sc GenS} (Algorithm~\ref{alg:spanning_complete}) is a specialized version of {\sc GenG} that leverages the complete connectivity of \( K_n \) to achieve constant amortized time per tree. A similar optimization can be applied to complete bipartite graphs \( K_{m,n} \), enabling generation in constant amortized time per tree.

\begin{theorem} \label{thm:bipartite_complexity}
Spanning trees of a complete bipartite graph $K_{m,n}$ with $m + n$ vertices can be generated in edge-exchange Gray code order in constant amortized time per tree, using $O((m+n)^2)$ space.
\end{theorem}
\begin{proof}
The proof adapts the {\sc GenS} algorithm to complete bipartite graphs \( K_{m,n} \). In \( K_{m,n} \), with partitions \( A \) and \( B \) containing $m$ and $n$ vertices respectively, the vertices in set \( P \) must belong to the same partition $A$ or $B$. 
We then partition \( C \) into two sets: one set can connect to \( P \), and the other cannot. Most operations follow those in {\sc GenS}, except for line 19 of Algorithm~\ref{alg:spanning_complete}, where a vertex \( c_i \) at level \( \ell \) cannot connect to another \( c_j \) at the same level due to the bipartite constraint. The strategy is to connect \( c_i \) to a child of a neighbor at level \( \ell \) if such a child exists.
To achieve constant amortized time per tree, we optimize finding valid connections using a \emph{linked hash table}, where each vertex stores its children in a hash table with linked lists. 
Implementation details of the linked hash table are similar to $C'$ discussed in Appendix \ref{sec:gens_opt}. 
Each vertex maintains the {linked hash table} with vertex indices as keys and linked lists of its children.
At each level \( \ell \), a {linked hash table} \( C' \) is constructed, containing vertices \( c_i \in C' \) at level \( \ell \) that connect to a vertex at level \( \ell-1 \) and have at least one child, with constant amortized time lookups and updates.
When progressing to level \( \ell+1 \), \( C' \) is updated in constant amortized time to maintain this property.
Now in case (b) when \( t_\ell \) sets \( a_i = 0 \) for a vertex \( c_i \) at level \( \ell \), \( c_i \) disconnects from its parent \( \pi(c_i) \) by removing the entry from the parent’s linked hash map.
If \( c_i \) has at least one child (detected in constant time via its linked hash table), a valid connection \( (c_a, c_b) \) is identified in constant time, where \( \pi(c_a) = c_i \) and \( c_b \) is the last vertex in \( C' \).
The {\sc Lift} operation takes constant time, as it modifies only \( c_i \) and \( c_a \).
Otherwise, if \( c_i \) has no children, a valid connection \( (c_i, c_j) \) is used, where \( \pi(c_j) \) is the last vertex in \( C' \) and can be found in constant time. 
\end{proof}

\newpage
Similar strategies can be applied to generate Gray codes for spanning trees of fan graphs and wheel graphs. 
\begin{theorem} \label{thm:fan_complexity}
Spanning trees of a fan graph $F_{n}$ with $n$ vertices can be generated in edge-exchange Gray code order in $O(n)$-amortized time per tree, using $O(n^2)$ space.
\end{theorem}

\begin{theorem} \label{thm:wheel_complexity}
Spanning trees of a wheel graph $W_{n}$ with $n$ vertices can be generated in edge-exchange Gray code order in $O(n)$-amortized time per tree, using $O(n^2)$ space.
\end{theorem}

\noindent However, the Gray codes produced for \( K_{m,n} \), \( F_n \), and \( W_n \) are edge-exchange Gray codes, not pivot Gray codes, similar to those for general graphs. 
Future research directions include developing pivot Gray codes for other special graph classes (for example, planar graphs would be interesting). Another avenue is to develop ranking and unranking algorithms for the spanning trees generated by {\sc GenG}, which may require ranking and unranking the underlying Gray codes for \( k \)-ary or mixed-radix strings. 

\section*{Acknowledgments.}
This research is supported by the Macao Polytechnic University research grant (Project code: RP/FCA-02/2022).


{\small
\bibliographystyle{siamplain}
\bibliography{myrefs}

\begin{thebibliography}{10}

\bibitem{Aichholzer2007}
{\sc O.~Aichholzer, F.~Aurenhammer, C.~Huemer, and B.~Vogtenhuber}, {\em Gray code enumeration of plane straight-line graphs}, Graphs Combin., 23 (2007), pp.~467--479.

\bibitem{behrooznia2025}
{\sc N.~Behrooznia and T.~M{\"u}tze}, {\em Listing spanning trees of outerplanar graphs by pivot-exchanges}, in 42nd International Symposium on Theoretical Aspects of Computer Science (STACS 2025), O.~Beyersdorff, M.~Pilipczuk, E.~Pimentel, and N.~K. Th{\'a}ng, eds., vol.~327, Schloss Dagstuhl -- Leibniz-Zentrum f{\"u}r Informatik, 2025, pp.~16:1--16:18.

\bibitem{Berger1967}
{\sc I.~Berger}, {\em The enumeration of trees without duplication}, IEEE Trans. Circuit Theory, 14 (1967), pp.~417--418.

\bibitem{cameron2024}
{\sc B.~Cameron, J.~Grubb, and J.~Sawada}, {\em Pivot {G}ray codes for the spanning trees of a graph ft. the fan}, Graphs Combin., 40 (2024).

\bibitem{Cayley1889}
{\sc A.~Cayley}, {\em A theorem on trees}, Q. J. Pure Appl. Math, 23 (1889), pp.~376--378.

\bibitem{Chakraborty2019}
{\sc M.~Chakraborty, S.~Chowdhury, J.~Chakraborty, R.~Mehera, and R.~K. Pal}, {\em Algorithms for generating all possible spanning trees of a simple undirected connected graph: an extensive review}, Complex Intell. Syst., 5 (2019), pp.~265--281.

\bibitem{Char1968}
{\sc J.~Char}, {\em Generation of trees, two-trees, and storage of master forests}, IEEE Trans. Circuit Theory, 15 (1968), pp.~228--238.

\bibitem{cummins1966}
{\sc R.~L. Cummins}, {\em Hamilton circuits in tree graphs}, IEEE Trans. Circuit Theory, 13 (1966), pp.~82--90.

\bibitem{feussner1902}
{\sc W.~Feussner}, {\em {\"U}ber {S}tromverzweigung in netzförmigen {L}eitern}, Ann. Phys., 314 (1902), pp.~1304--1329.

\bibitem{Gabow1978}
{\sc H.~Gabow and E.~Myers}, {\em Finding all spanning trees of directed and undirected graphs}, SIAM J. Comput., 7 (1978), pp.~280--287.

\bibitem{Hakimi1961}
{\sc S.~Hakimi}, {\em On trees of a graph and their generation}, J. Franklin Inst., 272 (1961), pp.~347--359.

\bibitem{holzmann19721}
{\sc C.~Holzmann and F.~Harary}, {\em On the tree graph of a matroid}, SIAM J. Appl. Math., 22 (1972), pp.~187--193.

\bibitem{kamae19777}
{\sc T.~Kamae}, {\em The existence of a {H}amilton circuit in a tree graph}, IEEE Trans. Circuits Syst., 24 (1977), pp.~716--718.

\bibitem{Kapoor1995}
{\sc S.~Kapoor and H.~Ramesh}, {\em Algorithms for enumerating all spanning trees of undirected and weighted graphs}, SIAM J. Comput., 24 (1995), pp.~247--265.

\bibitem{Kapoor2000}
{\sc S.~Kapoor and H.~Ramesh}, {\em An algorithm for enumerating all spanning trees of a directed graph}, Algorithmica, 27 (2000), pp.~120--130.

\bibitem{Katoh2009a}
{\sc N.~Katoh and S.~Tanigawa}, {\em Enumerating edge-constrained triangulations and edge-constrained non-crossing geometric spanning trees}, Discret. Appl. Math., 157 (2009), pp.~3569--3585.
\newblock 6th Int. Conf. Graphs Optim. 2007.

\bibitem{Katoh2009b}
{\sc N.~Katoh and S.~Tanigawa}, {\em Fast enumeration algorithms for non-crossing geometric graphs}, Discret. Comput. Geom., 42 (2009), pp.~443--468.

\bibitem{Kirchhoff1847}
{\sc G.~Kirchhoff}, {\em {\"U}ber die {A}ufl{\"o}sung der {G}leichungen, auf welche man bei der {U}ntersuchung der linearen {V}erteilung galvanischer {S}tr{\"o}me gef{\"u}hrt wird}, Ann. Phys., 72 (1847), pp.~497--508.

\bibitem{kishi1967}
{\sc G.~Kishi and Y.~Kajitani}, {\em On {H}amiltonian circuits in tree graphs}, IEEE Trans. Circuit Theory, 14 (1967), pp.~360--361.

\bibitem{knuth2011}
{\sc D.~Knuth}, {\em The Art of Computer Programming: Combinatorial Algorithms, Part 1}, Addison-Wesley Prof., 1st~ed., 2011.

\bibitem{li2009gray}
{\sc Y.~Li and J.~Sawada}, {\em Gray codes for reflectable languages}, Inf. Process. Lett., 109 (2009), pp.~296--300.

\bibitem{Matsui1997}
{\sc T.~Matsui}, {\em A flexible algorithm for generating all the spanning trees in undirected graphs}, Algorithmica, 18 (1997), pp.~530--543.

\bibitem{Mayeda1965}
{\sc W.~Mayeda and S.~Seshu}, {\em Generation of trees without duplications}, IEEE Trans. Circuit Theory, 12 (1965), pp.~181--185.

\bibitem{doi:10.1137/23M1612019}
{\sc A.~Merino and T.~M\"{u}tze}, {\em Traversing combinatorial 0/1-polytopes via optimization}, SIAM J. Comput., 53 (2024), pp.~1257--1292.

\bibitem{Merino2022}
{\sc A.~Merino, T.~M{\"u}tze, and A.~Williams}, {\em All your bases are belong to us: Listing all bases of a matroid by greedy exchanges}, in Proc. 11th Int. Conf. Fun Algorithms (FUN), P.~Fraigniaud and Y.~Uno, eds., vol.~226, 2022, pp.~22:1--22:28.

\bibitem{Minty1965}
{\sc G.~Minty}, {\em A simple algorithm for listing all the trees of a graph}, IEEE Trans. Circuit Theory, 12 (1965), pp.~120--120.

\bibitem{mutze2023combinatorial}
{\sc T.~M{\"u}tze}, {\em Combinatorial {G}ray codes — an updated survey}, Electron. J. Comb., 30 (2023), p.~3.

\bibitem{Naddef1981}
{\sc D.~Naddef and W.~Pulleyblank}, {\em Hamiltonicity and combinatorial polyhedra}, J. Combin. Theory Ser. B, 31 (1981), pp.~297--312.

\bibitem{OEIS2025A000272}
{\sc {OEIS Foundation Inc.}}, {\em Number of trees on $n$ labeled nodes.}
\newblock Entry A000272 in The On-Line Encyclopedia of Integer Sequences, 2025, \url{https://oeis.org/A000272}.

\bibitem{Prufer1918}
{\sc H.~Pr{\"u}fer}, {\em Neuer {B}eweis eines {S}atzes {\"u}ber {P}ermutationen}, Arch. Mat. Phys., 27 (1918), pp.~142--144.

\bibitem{ruskey1996combinatorial}
{\sc F.~Ruskey}, {\em Combinatorial generation}, Working version (1i),  (1996).

\bibitem{Sav97}
{\sc C.~Savage}, {\em A survey of combinatorial {Gray} codes}, SIAM Review,  (1997), pp.~605--629.

\bibitem{10.1007/978-3-030-85088-3_15}
{\sc J.~Sawada, A.~Williams, and D.~Wong}, {\em Inside the {B}inary {R}eflected {G}ray {C}ode: Flip-swap languages in 2-{G}ray code order}, in Combinatorics on Words, T.~Lecroq and S.~Puzynina, eds., Cham, 2021, pp.~172--184.

\bibitem{flipswap}
{\sc J.~Sawada, A.~Williams, and D.~Wong}, {\em Flip-swap languages in binary reflected {G}ray code order}, Theor. Comput. Sci., 933 (2022), pp.~138--148.

\bibitem{Schrijver2003}
{\sc A.~Schrijver}, {\em Combinatorial Optimization: Polyhedra and Efficiency}, vol.~B of Algorithms and Combinatorics, Springer-Verlag, Berlin, 2003.
\newblock Matroids, trees, stable sets, Chapters 39--69.

\bibitem{Shank1968}
{\sc H.~Shank}, {\em A note on {H}amilton circuits in tree graphs}, IEEE Trans. Circuit Theory, 15 (1968), p.~86.

\bibitem{Shioura1995}
{\sc A.~Shioura and A.~Tamura}, {\em Efficiently scanning all spanning trees of an undirected graph}, J. Oper. Res. Soc. Japan, 38 (1995), pp.~331--344.

\bibitem{Shioura1997}
{\sc A.~Shioura, A.~Tamura, and T.~Uno}, {\em An optimal algorithm for scanning all spanning trees of undirected graphs}, SIAM J. Comput., 26 (1997), pp.~678--692.

\bibitem{smith1997}
{\sc M.~Smith}, {\em Generating spanning trees}, master's thesis, University of Victoria, 1997.

\bibitem{Vajnovszki2007GrayCO}
{\sc V.~Vajnovszki}, {\em Gray code order for {L}yndon words}, Discret. Math. Theor. Comput. Sci., 9 (2007).

\bibitem{williams2013greedy}
{\sc A.~Williams}, {\em The greedy {G}ray code algorithm}, in Workshop on Algorithms and Data Structures, Springer, 2013, pp.~525--536.

\bibitem{Winter1985}
{\sc P.~Winter}, {\em An algorithm for the enumeration of spanning trees}, BIT Numer. Math., 26 (1985), pp.~44--62.

\end{thebibliography}
}

\newpage
\appendix
\section{Optimization for {\sc GenS}.}\label{sec:gens_opt}

In this section, we present optimizations for {\sc GenS} (Algorithm~\ref{alg:spanning_complete}) that enable the generation of pivot Gray codes for spanning trees of \( K_n \) in constant amortized time per tree, as detailed in Theorem~\ref{thm:gens_complexity}.

\textbf{(Line 8 -- 9) Computing the initial $t_\ell$.} The initial $t_\ell = b_1b_2 \cdots b_{\vert C\vert}$ describes the connections between level $\ell-1$ and level $\ell$. We construct a hash table $m$ with $t$ key-value pairs. The keys of the hash table $m$ are the indices $i$ of each $v_i \in V$. The values for each element in $m$ are set to 0 by default, and the value is updated by traversing $p_i \in P$, setting $m[p_i] \gets i$. Then we update each $b_i \gets m[\pi(c_i)]$ where $b_i$ indicates that $\pi(c_i) = p_{b_i}$. 

\textbf{(Line 19) Connect $c_i$ to a sibling $c_j$ at level $\ell$.}
A hash table $C^\prime$ chaining with a doubly linked list can be constructed for optimizing the operation of connecting a disconnected vertex $c_i$ to another vertex $c_j$ in level $\ell$. The structure of $C^\prime$ is shown in Figure \ref{fig:cprime}. $C^\prime$ has at most $\vert C \vert$ key-value pairs, where each key of $c_i$ points to a tuple $(-1, -1)$ initially. A pointer $C^\prime_{-1}$ pointing to the last activated vertex is set to $(-1)$ by default. For each $c_i \in C$ where $b_i > 0$, we activate $c_i$ by setting $C^\prime [i] \gets (C^\prime_{-1}[0], C^\prime_{-1})$ then $C^\prime_{-1}[0] \gets i$. 
When a vertex connects to a $p_i \in P$, we activate this vertex by the operation described above.
When a vertex $c_j$ disconnects from $\pi(c_j) \in P$, we link $C^\prime[C^\prime[j][0]]$ and $C^\prime[C^\prime[j][1]]$ to maintain the structure of doubly linked list and set $C^\prime[j] \gets (-1, -1)$ to deactivate $c_j$. Because the activated vertices in $C^\prime$ are in level $\ell$, so they must not be descendants of $c_j$. The valid link $(c_j, c_{C^\prime_{-1}})$ can be found in constant time. 

\begin{figure}[h]
    \centering
    \includegraphics[width=0.6\linewidth, trim={85 5 80 0}, clip]{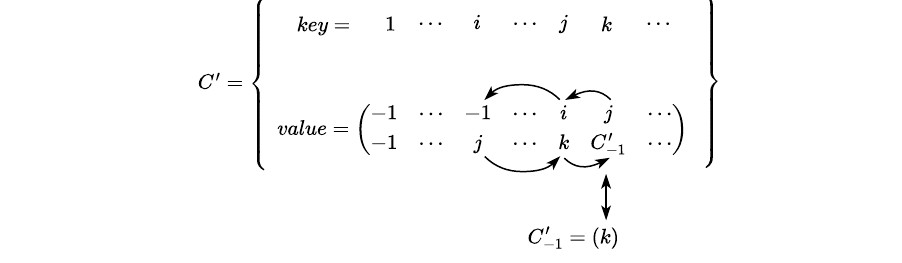}
    \caption{Structure of hash table $C^\prime$ chaining with a linked list.}
    \label{fig:cprime}
\end{figure}

\textbf{(Line 2 -- 6) Output spanning tree.}
The original {\sc GenS} outputs a spanning tree after traversing the input $V$ and ensuring $\vert C \vert = 0$, which takes $O(t)$ time complexity. We optimize the output operation leveraging $C^\prime$ mentioned above.
Note that $C^\prime$ contains the vertices in level $\ell$. We can set a variable for $\vert C^\prime \vert$ that records the number of activated vertices in $C^\prime$ can be updated after each connection and disconnection. When $\vert C \vert = \vert C^\prime \vert$, which indicates no vertex in the next level, {\sc GenS} prints the spanning tree directly, which skips the construction of $P$ and $C$ in the next level.

Pseudocode of the optimized algorithm is provided in Algorithm \ref{alg:spanning_complete_opt}. 
A complete Python implementation of the optimized algorithm is provided in Appendix \ref{sec2:gens_opt_code}.

\begin{algorithm}
    \caption{Optimized algorithm for generating pivot Gray codes for spanning trees of complete graphs in constant amortized time per spanning tree.}
    \label{alg:spanning_complete_opt}
    \begin{algorithmic}[1]
        \Procedure{\sc GenS}{$\ell, V, t_{\ell-1}, n_c$}
        \If{$n_c = 0$}
            \State{Print spanning tree}
            \State{\textbf{return}}
        \EndIf
        \State{$P \gets v_i \in V \mathrm{~where~} t_{\ell-1, i} > 0$}
        \State{$C \gets v_i \in V \mathrm{~where~} t_{\ell-1, i} = 0$}
        \State{$t_{\ell} \gets 0^{\vert C \vert}$}

        \State{$C^\prime \gets \{\}$}

        \State{$m \gets \{p_i \to i\}_{p_i \in P}$}
        \For{$c_i \in C$}
            \If{$\pi(c_i) \in m$}
                \State{$t_{\ell, i} \gets m[\pi(c_i)]$}
                \State{Append $i$ to $C^\prime$}
            \EndIf
        \EndFor
        \State{{\sc GenS}($\ell + 1, C, t_{\ell}, \vert C \vert - \vert C^\prime\vert$)}
        \For{{\sc NextK}$(t_{\ell}, \vert P \vert +1 ) \neq \emptyset$}
        \State{$t_{\ell}, \delta \gets ${\sc NextK}$(t_{\ell}, \vert P \vert +1 )$}
            \If{$\vert \delta \vert > 1$}
                \State{$\pi(c_{\delta_1}) \gets \pi(c_{\delta_0})$}
                \State{{\sc Lift}$(c_{\delta_1})$}
                \State{Remove $\delta_0$ from $C^\prime$}
                \State{Append $\delta_1$ to $C^\prime$}
            \ElsIf{$t_{\ell, \delta_0} > 0$}
                \State{$\pi(c_{\delta_0}) \gets p_{b_{\delta_0}} \in P$}
                \If{$\delta_0$ not in $C^\prime$}{~Append $\delta_0$ to $C^\prime$}
                \EndIf
            \Else
                \State{Remove $\delta_0$ from $C^\prime$}
                \State{$\pi(c_{\delta_0}) \gets c_{C^\prime_{-1}}$}
            \EndIf
            \State{{\sc GenS}($\ell + 1, C, t_{\ell}, \vert C \vert - \vert C^\prime\vert$)}
        \EndFor
        \EndProcedure
    \end{algorithmic}
\end{algorithm}
\clearpage
\subsection{Python code to generate pivot Gray codes for spanning trees of complete graphs.}\label{sec2:gens_opt_code}
~

\lstset{style=mystyle}
\begin{lstlisting}[language=Python]
from k_ary import k_ary  

class Node:
    def __init__(self, val):
        self.val = val
        self.parent = None

class LinkedList:
    def __init__(self):
        self.map = {}
        self.last = -1
        self.len = 0

    def append(self, val):
        self.map[val] = [self.last, -1]
        if self.last != -1:
            self.map[self.last][1] = val
        self.last = val
        self.len += 1

    def remove(self, val):
        val_last = self.map[val][0]
        val_next = self.map[val][1]
        if val_last != -1:
            self.map[val_last][1] = val_next
        if val_next != -1:
            self.map[val_next][0] = val_last
            self.last = val_next
        else:
            self.last = val_last
        del self.map[val]
        self.len -= 1
    
node_list = []

def lift_node(node):
    if node.parent is not None:
        lift_node(node.parent)
        node.parent.parent = node

def spanning_level(V, t_, n_c):
    if n_c == 0:
        yield node_list
        return

    P = [V[i] for i in range(len(V)) if t_[i] > 0]
    C = [V[i] for i in range(len(V)) if t_[i] == 0]
    t_l = [0 for _ in range(n_c)]
    C_ = LinkedList()

    m = {}
    for i in range(len(P)):
        m[P[i].val] = i+1

    for i in range(len(C)):
        if C[i].parent.val in m:
            t_l[i] = m[C[i].parent.val]
            C_.append(i)

    yield from spanning_level(C, t_l, len(C) - C_.len)

    K = k_ary([len(P) for _ in range(n_c)], t_l)
    next(K)

    for t_l, idx in K:
        if len(idx) > 1:
            temp = C[idx[0]].parent
            C[idx[0]].parent = None
            lift_node(C[idx[1]])
            C[idx[1]].parent = temp
            C_.append(idx[1])
            C_.remove(idx[0])

        elif t_l[idx[0]] > 0:
            C[idx[0]].parent = P[t_l[idx[0]]-1]
            if not idx[0] in C_.map:
                C_.append(idx[0])

        else:
            C_.remove(idx[0])
            C[idx[0]].parent = C[C_.last]

        yield from spanning_level(C, t_l, len(C) - C_.len)

if __name__ == '__main__':

    def str_tree(nodes):
        res = ''
        for nd in nodes:
            res += '{}->{}; '.format(nd.val+1, nd.parent.val+1 if nd.parent is not None else -1)

        return res

    print('Enter n:')
    n = int(input())

    prev = None
    for i in range(0, n):
        new_node = Node(i)
        if prev is not None:
            new_node.parent = prev
        node_list.append(new_node)
        prev = new_node
    print('==============================================')

    cnt = 0
    for t in spanning_level(node_list, [1] + [0 for _ in range(n - 1)], n - 1):
        print(str_tree(t))
        cnt += 1

    print('Total: ' + str(cnt))

\end{lstlisting}
\clearpage

\section{Python code to generate Gray codes for \texorpdfstring{$k$}{k}-ary strings and mixed-radix strings.}\label{sec:genk_code}
~


\lstset{style=mystyle}
\begin{lstlisting}[language=Python]
def k_ary(k_list, alpha):

    n = len(k_list)
    d = [1 if a != 1 else -1 for a in alpha]
    idx = []
    s = [sum(alpha)]

    order = [i for i in range(n)]
    for i in range(n-1, -1, -1):
        if k_list[i] > 1:
            order[n-1], order[i] = i, n-1
            break

    def _next(i):
        i_ = order[i]

        for r in range(k_list[i_] + 1):

            if i == n-1:
                if s[0] != 0:
                    if len(idx) > 1 and idx[0] == idx[1]:
                        idx.pop()
                    yield alpha, idx
                    idx[:] = []
            else:
                i_next = order[i + 1]
                yield from _next(i + 1)
                d[i_next] = 1 if alpha[i_next] == k_list[i_next] else -d[i_next]

            if r < k_list[i_]:
                a_b = alpha[i_]
                a_a = (alpha[i_] + d[i_] + k_list[i_] + 1) % (k_list[i_] + 1)
                s[0] += a_a - a_b
                alpha[i_] = a_a
                idx.append(i_)

    yield from _next(0)


\end{lstlisting}

\clearpage

\section{Python code to generate an edge-exchange Gray code for spanning trees of a general graph.}\label{sec:geng_code}
~


\lstset{style=mystyle}
\begin{lstlisting}[language=Python]
from k_ary import k_ary  
import numpy as np

class Node:
    def __init__(self, val):
        self.val = val
        self.parent = None

node_list = []
adj = [[0]]

def lift_node(node):
    if node.parent is not None:
        lift_node(node.parent)
        node.parent.parent = node

def get_connected_components(nodes):
    visited = {nd.val: False for nd in nodes}
    cc_list = []
    for nd in nodes:
        if not visited[nd.val]:
            visited[nd.val] = True
            group = [nd]
            dfs(nd, nodes, visited, group)
            cc_list.append(group)
    return cc_list

def dfs(cnode, nodes, visited, group):
    for nd in nodes:
        if not visited[nd.val] and adj[cnode.val][nd.val] == 1:
            visited[nd.val] = True
            group.append(nd)
            dfs(nd, nodes, visited, group)

def check_k(parents, children):
    c = [0 for _ in range(len(children))]
    p_list = [[] for _ in range(len(children))]
    for i_c in range(len(children)):
        n_p = 0
        for i_p in range(len(parents)):
            if adj[parents[i_p].val][children[i_c].val] == 1:
                p_list[i_c].append(parents[i_p])
                n_p += 1
                if children[i_c].parent.val == parents[i_p].val:
                    c[i_c] = n_p            
    k_list = [len(p) for p in p_list]
    return c, k_list, p_list

def get_connection(cnode, nodes):
    m = {nd.val: [] for nd in nodes}
    for nd in nodes:
        try:
            m[nd.parent.val].append(nd)
        except KeyError:
            pass

    nodes_under_c = [cnode]
    temp_nodes = [cnode]
    while len(temp_nodes) > 0:
        new_nodes = []
        for tnd in temp_nodes:
            new_nodes += m[tnd.val]
            del m[tnd.val]
        nodes_under_c += new_nodes
        temp_nodes = new_nodes

    nodes_other = [nd for nd in nodes
                   if m.get(nd.val) is not None and nd.val != cnode.val]

    for nd1 in nodes_under_c:
        for nd2 in nodes_other:
            if adj[nd1.val][nd2.val] == 1:
                return nd1, nd2

def spanning_level(V, t_):
    C = [V[i] for i in range(len(V)) if t_[i] == 0]
    if len(C) == 0:
        yield node_list
        return

    P = [V[i] for i in range(len(V)) if t_[i] > 0]
    C_ = get_connected_components(C)
    F = [check_k(P, C_sub) for C_sub in C_]

    def spanning_subtree(i_tree):
        if i_tree >= len(F):
            yield from spanning_level(sum(C_, []), sum([g[0] for g in F], []))
            return

        yield from spanning_subtree(i_tree+1)

        K = k_ary(F[i_tree][1].copy(), F[i_tree][0].copy())
        next(K)

        for t_l_i, idx in K:
            if len(idx) > 1:
                C_[i_tree][idx[0]].parent = None
                lift_node(C_[i_tree][idx[1]])
                C_[i_tree][idx[1]].parent = F[i_tree][2][idx[1]][t_l_i[idx[1]] - 1]

            elif t_l_i[idx[0]] > 0:
                C_[i_tree][idx[0]].parent = F[i_tree][2][idx[0]][t_l_i[idx[0]] - 1]

            else:
                nd1, nd2 = get_connection(C_[i_tree][idx[0]], C_[i_tree])
                C_[i_tree][idx[0]].parent = None
                lift_node(nd1)
                nd1.parent = nd2

            F[i_tree][0][:] = t_l_i[:]
            yield from spanning_subtree(i_tree+1)

    yield from spanning_subtree(0)

def get_first_spanning_tree():
    n = len(adj)
    check_list = [True for _ in range(n)]
    node_list[:] = [Node(i) for i in range(n)]
    check_list[0] = False
    for i in range(n):
        for j in range(n):
            if check_list[j] and adj[i][j] == 1:
                check_list[j] = False
                node_list[j].parent = node_list[i]

    return node_list

if __name__ == '__main__':
    def str_tree(nodes, vmap=None):
        res = ''
        for nd in nodes:
            res += '{}->{}; '.format(nd.val+1, nd.parent.val+1 if nd.parent is not None else -1)
        return res

    print('Enter type: #1 complete graph, #2 fan graph, #3 wheel graph, #4 petersen, #5 custom graph')
    i_type = int(input())

    if i_type != 4:
        print('Enter n:')
        n = int(input())
    else:
        n = 10

    if i_type == 1:
        adj = [[1 for _ in range(n)] for _ in range(n)]
        
    elif i_type == 2 or i_type == 3:
        # fan graph
        adj = np.zeros((n, n), dtype=int)
        adj[0, :] = 1
        adj[:, 0] = 1
        for i in range(1, n):
            adj[i, max(0, i - 1):min(n, i + 2)] = 1
            adj[max(0, i - 1):min(n, i + 2), i] = 1

        if i_type == 3:
            # wheel graph
            adj[1, -1] = 1
            adj[-1, 1] = 1

    elif i_type == 4:
        # petersen graph
        adj = np.array([
            [0, 1, 0, 0, 1, 1, 0, 0, 0, 0],
            [1, 0, 1, 0, 0, 0, 1, 0, 0, 0],
            [0, 1, 0, 1, 0, 0, 0, 1, 0, 0],
            [0, 0, 1, 0, 1, 0, 0, 0, 1, 0],
            [1, 0, 0, 1, 0, 0, 0, 0, 0, 1],
            [1, 0, 0, 0, 0, 0, 0, 1, 1, 0],
            [0, 1, 0, 0, 0, 0, 0, 0, 1, 1],
            [0, 0, 1, 0, 0, 1, 0, 0, 0, 1],
            [0, 0, 0, 1, 0, 1, 1, 0, 0, 0],
            [0, 0, 0, 0, 1, 0, 1, 1, 0, 0],
        ])
        
    elif i_type == 5:
        print('Enter edges: e.g., 1,2; 2,3')
        edge_str = input().replace(' ', '').split(';')
        edge_list = []
        for es in edge_str:
            if es:
                edge_list.append(tuple(map(int, es.split(','))))
        adj = np.zeros((n, n), dtype=int)
        for edge in edge_list:
            adj[edge[0]-1, edge[1]-1] = 1
            adj[edge[1]-1, edge[0]-1] = 1

    else:
        raise Exception('Invalid input')

    get_first_spanning_tree()  # bfs

    cnt = 0
    for t in spanning_level(node_list, [1] + [0 for _ in range(len(adj)-1)]):
        print(str_tree(t))
        cnt += 1
    print('Total: ' + str(cnt))

\end{lstlisting}

\clearpage

\section{Pivot Gray code of \texorpdfstring{$K_5$}{K5} generated by {\sc GenS}.}\label{sec:gens_k5}
~
\begin{figure}[H]
    \centering
    \includegraphics[width=0.89\linewidth, trim={0 20 0 0}, clip]{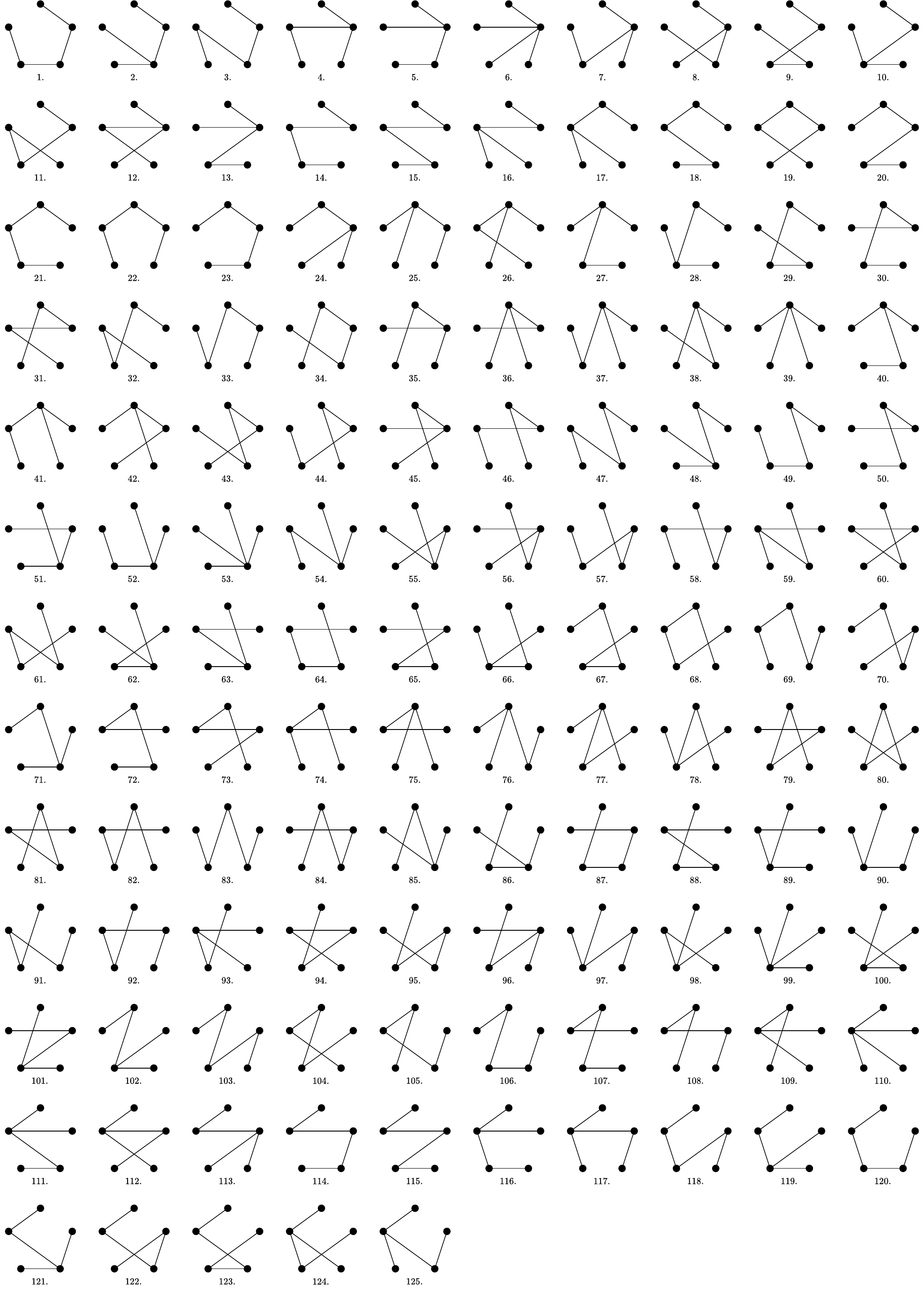}
\end{figure}

\section{Edge-exchange Gray code of the Petersen graph generated by {\sc GenG} in compact representation}\label{sec:geng_petersen}
~

{\scriptsize\noindent
{123112345, 123112a45, 123112a75, 123112375, 129112375, 129112a75, 129112675, 123112645, 123112675, 123112677, 123112678, 123112648, 123112647, 129112677, 129112678, 129112378, 129112a77, 129112377, 123112347, 123112a47, 123112348, 123112378, 123112a77, 123112377, 123119347, 123119348, 12311a348, 12311a378, 123119a47, 12911a378, 12911a678, 123119647, 123119648, 12311a648, 12311a678, 12311a675, 12311a645, 123119645, 12911a675, 12911a375, 12911aa75, 123119a45, 123119345, 12311a345, 12311a375, 12311aa75, 12311aa45, 12311aa65, 12311a365, 123119365, 123119a65, 129119a65, 129119365, 12911a365, 12911aa65, 12911a665, 129119665, 123119665, 12311a665, 12311a668, 123119667, 123119668, 129119668, 129119667, 12911a668, 12911a368, 129119a67, 129119367, 129119368, 123119368, 123119367, 123119a67, 12311a368, 123112368, 123112a67, 123112367, 129112367, 129112a67, 129112368, 129112668, 129112667, 123112667, 123112668, 123112665, 129112665, 129112365, 129112a65, 123112a65, 123112365, 125112365, 125112a65, 125112665, 125112667, 125112668, 125112368, 125112a67, 125112367, 125119367, 125119368, 125119a67, 12511a368, 12511a668, 125119667, 125119668, 125119665, 12511a665, 12511a365, 12511aa65, 125119a65, 125119365, 125119345, 125119a45, 12511aa45, 12511a345, 12511a375, 12511aa75, 12511a675, 125119645, 12511a645, 12511a648, 125119648, 125119647, 12511a678, 12511a378, 125119a47, 125119347, 125119348, 12511a348, 125112348, 125112a47, 125112347, 125112377, 125112a77, 125112378, 125112678, 125112677, 125112647, 125112648, 125112645, 125112675, 125112375, 125112a75, 125112a45, 125112345, 145112345, 145112375, 145112a75, 145112a45, 185112a45, 185112a75, 185112675, 185112645, 145112645, 145112675, 145112677, 145112678, 145112648, 145112647, 185112647, 185112648, 185112678, 185112677, 185112a77, 185112a47, 145112348, 145112378, 145112377, 145112a77, 145112a47, 145112347, 145119347, 145119348, 14511a348, 145119a47, 14511a378, 185119a47, 185119647, 185119648, 18511a648, 18511a678, 14511a678, 145119647, 145119648, 14511a648, 14511a645, 145119645, 14511a675, 18511a675, 185119645, 18511a645, 18511aa45, 185119a45, 18511aa75, 14511a375, 145119345, 14511a345, 14511aa45, 145119a45, 14511aa75, 14511aa65, 145119a65, 145119365, 14511a365, 18511aa65, 185119a65, 185119665, 18511a665, 14511a665, 145119665, 145119667, 145119668, 14511a668, 18511a668, 185119667, 185119668, 145119368, 145119367, 145119a67, 14511a368, 185119a67, 185112a67, 145112368, 145112367, 145112a67, 145112667, 145112668, 185112668, 185112667, 185112665, 145112665, 145112365, 145112a65, 185112a65, 183112a65, 149112365, 149112a65, 189112a65, 189112665, 183112665, 149112665, 149112667, 149112668, 183112668, 183112667, 189112667, 189112668, 149112368, 149112367, 149112a67, 189112a67, 183112a67, 183119a67, 14911a368, 149119368, 149119367, 149119a67, 189119a67, 189119667, 189119668, 18911a668, 18311a668, 183119667, 183119668, 149119668, 149119667, 14911a668, 14911a665, 149119665, 183119665, 18311a665, 18911a665, 189119665, 189119a65, 149119a65, 14911aa65, 18911aa65, 14911a365, 149119365, 183119a65, 18311aa65, 18311aa45, 18311aa75, 18911aa75, 14911aa75, 14911a375, 183119a45, 183119645, 18311a645, 18311a675, 18911a675, 14911a675, 14911a678, 183119647, 183119648, 18311a648, 18311a678, 18911a678, 189112a77, 149112a77, 149112377, 183112a77, 149112378, 183112a47, 183112647, 183112648, 183112678, 189112678, 189112677, 183112677, 149112677, 149112678, 149112675, 183112645, 183112675, 189112675, 189112a75, 183112a75, 149112a75, 149112375, 183112a45, 183192a45, 183182a45, 183182a65, 189182a65, 149182a65, 149192675, 149192375, 189192675, 149182375, 183192675, 183192a75, 183182a75, 189182a75, 189192a75, 149192a75, 149182a75, 149182a77, 189182a77, 189192a77, 149192a77, 149192377, 149192677, 189192677, 149182377, 183192677, 183192a77, 183182a77, 183192678, 149182378, 149192378, 149192678, 189192678, 189182a67, 183182a47, 183182a67, 183192a47, 149182a67, 14918aa65, 14919aa75, 14918aa75, 18918aa75, 18318aa75, 18319aa75, 18919aa75, 18319aa45, 18318aa45, 18318aa65, 18918aa65, 18919a675, 14919a375, 14919a675, 14918a375, 18319a675, 183189a65, 183189a45, 189189a65, 183199a45, 149189a65, 145189a65, 14518aa65, 14518aa75, 14519aa75, 14519aa45, 14518aa45, 145189a45, 145199a45, 145199345, 145189345, 145199645, 14519a645, 14518a345, 14519a345, 14519a375, 14518a375, 14518a365, 14519a675, 145189365, 18519a675, 18519a645, 185199645, 185199a45, 185189a45, 18518aa45, 18519aa45, 18519aa75, 18518aa65, 18518aa75, 185189a65, 145189368, 14518a378, 14518a368, 14519a378, 145189367, 145189347, 145189348, 145199348, 145199347, 14519a348, 14518a348, 14519a648, 145199648, 145199647, 145199a47, 145189a47, 185189a47, 185199a47, 185199648, 185199647, 18519a648, 185192648, 185192678, 185192677, 185182a77, 185192a77, 185182a67, 185182a47, 185192a47, 185192647, 145192647, 145192347, 145192a47, 145182a47, 145182347, 145182367, 145182a67, 145182a77, 145182377, 145192377, 145192677, 145192a77, 145192378, 145182378, 145192678, 145182368, 145182348, 145192348, 145192648, 145192645, 145182345, 145192345, 145182365, 145182375, 145192375, 145192675, 145192a75, 145182a75, 145182a65, 145182a45, 145192a45, 185192a45, 185182a45, 185182a65, 185182a75, 185192a75, 185192675, 185192645, 125192645, 125192345, 125182345, 125182a45, 125192a45, 125182a65, 125182365, 125182375, 125192375, 125192675, 125192a75, 125182a75, 125182a77, 125192677, 125192a77, 125182a67, 125182a47, 125192647, 125192a47, 125192347, 125182347, 125182367, 125182377, 125192377, 125192378, 125182378, 125182368, 125182348, 125192348, 125192648, 125192678, 125199648, 125199a47, 125199647, 125189a47, 12519a648, 12519a348, 12518a348, 125189347, 125199347, 125199348, 125189348, 125189368, 12518a368, 12518a378, 12519a378, 125189367, 125189365, 125189a65, 12518aa65, 12519aa75, 12518aa75, 12519a675, 12519a375, 12518a375, 12518a365, 12518a345, 12519a345, 12519a645, 12519aa45, 12518aa45, 125189a45, 125199a45, 125199645, 125199345, 125189345, 123189345, 123199345, 123189365, 123199645, 123199a45, 123189a45, 123189a65, 12318aa65, 12318aa45, 12318aa75, 12319aa75, 12319aa45, 12319a645, 12319a675, 12319a375, 12319a345, 12318a345, 12318a365, 12318a375, 12918a375, 12919a375, 12918a365, 12919a675, 12919aa75, 12918aa75, 12918aa65, 129189a65, 129189365, 129189367, 12919a378, 12918a378, 129189368, 12918a368, 12319a648, 123189a47, 123199a47, 123199647, 123199648, 123199348, 12319a348, 12318a348, 123189348, 123189347, 123199347, 123189367, 123189368, 12318a368, 12318a378, 12319a378, 123192378, 123182378, 123182377, 123192377, 123192677, 123192a77, 123182a77, 123192678, 123192648, 123192647, 123192a47, 123182a47, 123182a67, 123182367, 123192347, 123182347, 123182348, 123192348, 123182368, 129182368, 129182367, 129182a67, 129192678, 129192677, 129192a77, 129182a77, 129182377, 129192377, 129192378, 129182378, 129182375, 129192375, 129192675, 129192a75, 129182a75, 129182a65, 129182365, 123182365, 123192345, 123182345, 123192645, 123192a45, 123182a45, 123182a65, 123182a75, 123192a75, 123192675, 123192375, 123182375, 123482375, 123492375, 123492378, 123482378, 123a82378, 123a92378, 123a92348, 123a82368, 123a82348, 123482345, 123482365, 123492345, 123492348, 123482368, 123482348, 123482a45, 123492a45, 123492645, 123492648, 123482a65, 123a92648, 123a92678, 123482a75, 123492a75, 123492675, 123492678, 123492677, 123482a77, 123492a77, 123a92a77, 123a82a77, 123a92677, 123a92647, 123a82a67, 123a82a47, 123a92a47, 123492a47, 123482a47, 123482a67, 123492647, 123492347, 123482367, 123482347, 123a82347, 123a82367, 123a92347, 123a92377, 123a82377, 123482377, 123492377, 129492377, 129482377, 129a82377, 129a92377, 125a92377, 125a82377, 125a82347, 125a82367, 129a82367, 125a92347, 129482367, 129482a67, 125a92647, 125a92a47, 125a82a47, 125a82a67, 129a82a67, 129a82a77, 129a92a77, 129a92677, 129492677, 129482a77, 129492a77, 125a92a77, 125a82a77, 125a92677, 125a92678, 129482a75, 129492a75, 129492675, 129492678, 129a92678, 129a82368, 125a82368, 125a82348, 129482368, 125a92348, 129482365, 129482375, 129492375, 129492378, 129a92378, 129a82378, 129482378, 125a82378, 125a92378, 125a9a378, 125a9a348, 129a9a378, 125a99348, 12949a378, 12948a378, 129489368, 12948a368, 129a8a368, 125a8a368, 125a89368, 129a89368, 125a89348, 125a8a348, 125a8a378, 129a8a378, 129a89367, 129489365, 129489367, 12948a365, 125a89367, 12349a675, 12349aa75, 12348aa75, 12348aa65, 123489a65, 123489a45, 123499a45, 123499645, 12349a645, 12349aa45, 12348aa45, 12349a648, 123499648, 123499a47, 123499647, 123489a47, 123a89a47, 123a99a47, 123a99647, 123a99648, 123a9a648, 123a9a348, 123a99348, 123a89348, 123a8a348, 123a8a378, 123a8a368, 123a89368, 123a9a378, 123a89367, 123a89347, 123a99347, 123499347, 12349a345, 123499345, 123489345, 12348a345, 123489347, 123489367, 12348a375, 12348a365, 123489365, 12349a375, 12349a378, 12348a378, 12348a368, 123489368, 123489348, 12348a348, 12349a348, 123499348, 149492378, 149492375, 149492a75, 149a92378, 149a92678, 149492678, 149492675, 189492675, 189492678, 189a92678, 189492a75, 189482a75, 149a82378, 149482378, 149482375, 149482a75, 145a92678, 183492675, 183492678, 183a92678, 185a92678, 185a82a67, 185a82a47, 185a92a47, 183a92a47, 183a82a47, 183a82a67, 189a82a67, 149a82a67, 145a92a47, 145a82a47, 145a82a67, 145a82367, 145a82347, 145a92347, 145a92647, 185a92647, 189482a67, 183482a47, 183482a67, 183492a47, 149482a67, 149482a77, 189482a77, 189492a77, 149492a77, 149492377, 149492677, 189492677, 149482377, 149a82377, 149a92377, 149a92677, 189a92677, 189a92a77, 149a92a77, 149a82a77, 189a82a77, 183a82a77, 185a82a77, 145a82a77, 145a92a77, 183a92a77, 185a92a77, 185a92677, 145a92677, 145a92377, 183a92677, 145a82377, 183492677, 183492a77, 183482a77, 183412a77, 145a12377, 145a12a77, 183a12a77, 185a12a77, 189a12a77, 149a12a77, 149a12377, 149412377, 149412a77, 189412a77, 183412a47, 145a12347, 145a12a47, 183a12a47, 185a12a47, 189412a75, 149a12378, 149412378, 149412375, 149412a75, 149412675, 149412677, 149a12677, 145a12677, 145a12678, 149412678, 149a12678, 145a12648, 145a12647, 185a12647, 183412647, 183a12647, 183412645, 183412648, 183a12648, 185a12648, 185a12678, 183a12678, 189a12678, 189412678, 183412678, 183412675, 189412675, 189412677, 183412677, 183a12677, 185a12677, 189a12677, 189a12667, 189412667, 189412665, 189412668, 189a12668, 183a12668, 183412668, 185a12668, 183412665, 183412667, 183a12667, 185a12667, 145a12667, 145a12668, 149a12668, 149412668, 149412665, 149412667, 149a12667, 149a12367, 149412367, 149412a67, 149a12a67, 189a12a67, 189412a67, 189412a65, 149a12368, 149412368, 149412365, 149412a65, 145a12a67, 183a12a67, 185a12a67, 145a12367, 183412a67, 183419a67, 145a19367, 145a19a67, 183a19a67, 185a19a67, 18941aa65, 149a1a368, 14941a368, 14941a365, 14941aa65, 149419a65, 149419368, 149419365, 149a19368, 149a19367, 149a19a67, 149419a67, 149419367, 189419a67, 189a19a67, 189419a65, 189419665, 189419667, 189a19667, 18941a665, 18941a668, 189a1a668, 189a19668, 189419668, 183419668, 183a19668, 185a19668, 185a1a668, 18341a668, 183a1a668, 18341a665, 183419665, 183419667, 183a19667, 185a19667, 145a19667, 145a1a668, 145a19668, 149a19668, 149419668, 14941a668, 149a1a668, 14941a665, 149419665, 149419667, 149a19667, 183a19647, 18341a675, 18341a645, 183419645, 183419647, 183419648, 18341a648, 18341a678, 183a1a678, 183a1a648, 183a19648, 185a19648, 185a1a648, 185a1a678, 189a1a678, 18941a678, 14941a678, 145a19648, 145a1a648, 145a1a678, 149a1a678, 123a19a47, 12341aa75, 12341aa45, 123419a45, 123419a47, 123419347, 12341a345, 123419345, 12341a375, 12341a378, 12341a348, 123419348, 123a19348, 123a1a348, 123a1a378, 123a19347, 129a1a378, 125a1a348, 125a1a378, 125a19348, 12941a378, 12941a678, 125a19648, 125a1a648, 125a1a678, 129a1a678, 123a19647, 12341a675, 12341a645, 123419645, 123419647, 123419648, 12341a678, 12341a648, 123a1a648, 123a1a678, 123a19648, 123a19668, 123419668, 125a19668, 129a19668, 129419668, 129419665, 129419667, 129a19667, 125a19667, 123a19667, 123419665, 123419667, 12341a665, 12941a665, 12941a668, 129a1a668, 125a1a668, 123a1a668, 12341a668, 12341a368, 12341a365, 123a1a368, 125a1a368, 129a1a368, 12941a365, 12941a368, 129419368, 129419365, 129419367, 129a19367, 129a19368, 125a19368, 125a19367, 123a19367, 123a19368, 123419368, 123419365, 123419367, 123419a67, 123419a65, 123a19a67, 125a19a67, 129a19a67, 129419a65, 129419a67, 12941aa65, 12341aa65, 123412a65, 129412a65, 129412a67, 129a12a67, 125a12a67, 123a12a67, 123412a67, 123412367, 123a12367, 125a12367, 129a12367, 129412367, 129412365, 129412368, 129a12368, 125a12368, 123a12368, 123412365, 123412368, 123412668, 123a12668, 123412665, 123412667, 123a12667, 125a12667, 125a12668, 129a12668, 129412668, 129412665, 129412667, 129a12667, 129a12677, 129412677, 125a12677, 129412675, 129412678, 129a12678, 125a12678, 125a12648, 125a12647, 123a12647, 123412647, 123412645, 123412648, 123a12648, 123a12678, 123412678, 123412675, 123412677, 123a12677, 123a12377, 123412377, 123412347, 123a12347, 123a12a47, 123412a47, 123412a77, 123a12a77, 123412a75, 123412a45, 123412345, 123412348, 123a12348, 123a12378, 123412375, 123412378, 129412378, 129a12378, 129412375, 125a12378, 125a12348, 129412a75, 129412a77, 129a12a77, 125a12a77, 125a12a47, 125a12347, 125a12377, 129412377, 129a12377, 389a12677, 385a12677, 389412677, 385a12647, 389412675, 383412675, 383412677, 383412647, 383412645, 383a12647, 383a12677, 383a19647, 383419647, 383419645, 38341a645, 38341a675, 78341a675, 78341a645, 783419647, 783419645, 783a19647, 783a19648, 783419648, 785a19648, 785a1a648, 789a1a678, 785a1a678, 78941a678, 78341a678, 78341a648, 783a1a648, 783a1a678, 383a1a678, 383a1a648, 38341a648, 38341a678, 38941a678, 389a1a678, 385a1a678, 385a1a648, 385a12648, 385a12678, 389a12678, 385a19648, 389412678, 383412678, 383412648, 383419648, 383a19648, 383a12678, 383a12648, 345a12648, 345a19648, 745a19648, 345a12678, 725a19648, 725a1a648, 745a1a648, 745a1a678, 725a1a678, 723a1a678, 729a1a678, 749a1a678, 723a1a648, 349a1a678, 345a1a678, 345a1a648, 34941a678, 72341a648, 72341a678, 72941a678, 74941a678, 749419a67, 749419367, 749a19367, 749a19a67, 749a19368, 749419368, 749419365, 749419a65, 729419a65, 729419365, 729419368, 789419a65, 729a19368, 729a19367, 729a19a67, 789a19a67, 789419a67, 729419367, 729419a67, 389419a67, 389a19a67, 389419a65, 349419a65, 349419368, 349419365, 349a19368, 349a19367, 349a19a67, 349419a67, 349419367, 349412367, 349412368, 34941a368, 349412a67, 349412a65, 34941aa65, 34941a365, 349412365, 74941a365, 74941aa65, 74941a368, 749a1a368, 349a12a67, 349a12367, 349a12368, 349a1a368, 38941aa65, 72941aa65, 78941aa65, 389412a65, 72941a365, 725a19367, 723a19367, 723a19a67, 725a19a67, 723419a67, 723419367, 783419a67, 783a19a67, 785a19a67, 745a19a67, 745a19367, 723419365, 723419368, 723a19368, 723419a65, 725a19368, 385a19a67, 383a19a67, 345a19a67, 345a19367, 383419a67, 383419667, 383a19667, 383419665, 385a19667, 785a19667, 783419665, 783419667, 783a19667, 383a12667, 383412667, 383412665, 38341a665, 385a12667, 78341a665, 78341a668, 783a1a668, 785a1a668, 385a12668, 383412668, 383a12668, 383a1a668, 38341a668, 385a1a668, 385a19668, 383419668, 383a19668, 783a19668, 783419668, 785a19668, 789a19668, 789419668, 389419668, 389a19668, 389a12668, 389412668, 38941a668, 389a1a668, 789a1a668, 78941a668, 78941a665, 389a12667, 389412667, 389412665, 38941a665, 389419665, 389419667, 389a19667, 789a19667, 789419665, 789419667, 729419667, 729419665, 729a19667, 749a19667, 749419665, 749419667, 349419667, 349419665, 349a19667, 349a12667, 349412667, 349412665, 34941a665, 74941a665, 72941a665, 72941a668, 729a1a668, 349a12668, 349412668, 34941a668, 74941a668, 749a1a668, 349a1a668, 349a19668, 349419668, 729419668, 729a19668, 749a19668, 749419668, 723419668, 723a19668, 725a19668, 745a19668, 345a19668, 345a12668, 345a1a668, 723a1a668, 725a1a668, 745a1a668, 72341a668, 723419667, 723a19667, 725a19667, 745a19667, 723419665, 345a19667, 34519a675, 34518a365, 34518a375, 34519a375, 74519a375, 74518a375, 74518a365, 74519a675, 345192675, 345192375, 345182375, 345182365, 345189365, 745189365, 745189345, 745199345, 745199645, 345192645, 345182345, 345192345, 345199345, 345189345, 345199645, 34519a645, 34518a345, 34519a345, 74519a345, 74518a345, 74519a645, 74519aa45, 74518aa45, 34518aa45, 34519aa45, 345192a45, 345182a45, 345189a45, 345199a45, 745199a45, 745189a45, 745189a65, 345192a75, 345182a75, 345182a65, 345189a65, 34518aa65, 34518aa75, 34519aa75, 74519aa75, 74518aa65, 74518aa75, 72518aa75, 72518aa65, 72519aa75, 78519aa75, 78518aa65, 78518aa75, 38518aa75, 38518aa65, 38519aa75, 385192a75, 385182a75, 385182a65, 385189a65, 785189a65, 725189a65, 725189a45, 725199a45, 385192a45, 385182a45, 385189a45, 785189a45, 785199a45, 385199a45, 38519aa45, 38518aa45, 72518aa45, 72519aa45, 78519aa45, 78518aa45, 72518a345, 72519a345, 72519a645, 78519a645, 38519a645, 385192645, 385199645, 725199345, 725199645, 785199645, 725189345, 72518a375, 72519a375, 72519a675, 78519a675, 72518a365, 38519a675, 385199647, 385199648, 385199a47, 785199a47, 725199347, 725199a47, 725199348, 725199648, 785199648, 785199647, 725199647, 725189347, 725189a47, 785189a47, 725189348, 385189a47, 38519a648, 385192648, 78519a648, 385192647, 72519a648, 345182a67, 345182a77, 345192a77, 345192677, 345192678, 345192378, 345182378, 345182368, 345182367, 345182377, 345192377, 345189367, 345189368, 34518a378, 34518a368, 34519a378, 74519a378, 74518a378, 74518a368, 745189368, 745189367, 745189347, 745189348, 745199348, 745199347, 745199a47, 745199647, 745199648, 745189a47, 74519a648, 74519a348, 74518a348, 34518a348, 345182347, 345182348, 345192348, 345192347, 34519a348, 34519a648, 345192a47, 345192647, 345192648, 345182a47, 345189a47, 345199a47, 345199647, 345199648, 345199348, 345199347, 345189347, 345189348, 383189a45, 383182a75, 383182a45, 389182a75, 389182a65, 389189a65, 383189a65, 383182a65, 783189a65, 789189a65, 783189a45, 783199a45, 389192a75, 383192a75, 383192a45, 383199a45, 349189a65, 729189a65, 749189a65, 349182a65, 723189a65, 72319a675, 72319a375, 72318a375, 72918a375, 72919a375, 72919a675, 74919a675, 74919a375, 78919a675, 74918a375, 78319a675, 72318a365, 72319a345, 72318a345, 72319a645, 72918a365, 38919a675, 34919a375, 34919a675, 34918a375, 38319a675, 38319aa75, 38919aa75, 38918aa75, 38318aa75, 38318aa45, 38318aa65, 38918aa65, 38319aa45, 34918aa65, 34918aa75, 34919aa75, 74919aa75, 72319aa75, 72919aa75, 72918aa75, 72318aa75, 74918aa75, 74918aa65, 72318aa45, 72318aa65, 72918aa65, 72319aa45, 78319aa45, 78318aa45, 78318aa65, 78918aa65, 78918aa75, 78318aa75, 78319aa75, 78919aa75, 78911aa75, 78311aa75, 78311aa45, 72311aa45, 72311aa75, 72911aa75, 74911aa75, 34911aa75, 38311aa45, 38311aa75, 38911aa75, 34911a375, 72311a345, 72311a375, 72911a375, 74911a375, 783119a45, 389112a75, 383112a75, 383112a45, 383119a45, 383119645, 383112675, 383112645, 389112675, 783119645, 78311a645, 78311a675, 78911a675, 38911a675, 38311a645, 38311a675, 72311a675, 72911a675, 74911a675, 72311a645, 34911a675, 34911a678, 72311a648, 72311a678, 72911a678, 74911a678, 783119647, 389112677, 383112677, 383112647, 383119647, 383119648, 383112678, 383112648, 389112678, 38911a678, 38311a678, 38311a648, 78311a648, 78311a678, 78911a678, 783119648, 749119368, 749119367, 749119a67, 729119a67, 729119367, 789119a67, 729119368, 389119a67, 349119a67, 349119367, 349119368, 349112368, 349112367, 34911a368, 349112a67, 74911a368, 783119a67, 723119367, 723119a67, 723119368, 383119a67, 383119667, 783119667, 783119668, 383119668, 383112668, 38311a668, 78311a668, 383112667, 389112667, 389112668, 38911a668, 78911a668, 789119668, 389119668, 389119667, 789119667, 729119667, 749119667, 349119667, 349119668, 729119668, 749119668, 74911a668, 34911a668, 349112668, 72911a668, 349112667, 72311a668, 723119668, 723119667, 723119665, 72311a665, 72911a665, 74911a665, 34911a665, 349112665, 349119665, 729119665, 749119665, 789119665, 389119665, 389112665, 38911a665, 78911a665, 78311a665, 38311a665, 383112665, 383119665, 783119665, 783119a65, 723119a65, 383119a65, 383112a65, 38311aa65, 72311aa65, 78311aa65, 72311a365, 723119365, 729119365, 749119365, 349119365, 349112365, 34911a365, 72911a365, 74911a365, 74911aa65, 72911aa65, 78911aa65, 38911aa65, 34911aa65, 349112a65, 389112a65, 389119a65, 349119a65, 749119a65, 729119a65, 789119a65, 785119a65, 725119a65, 385119a65, 385112a65, 38511aa65, 72511aa65, 78511aa65, 72511a365, 725119365, 745119365, 345119365, 345112365, 34511a365, 74511a365, 74511aa65, 34511aa65, 345112a65, 345119a65, 745119a65, 745119665, 345119665, 345112665, 34511a665, 74511a665, 72511a665, 725119665, 785119665, 385119665, 385112665, 38511a665, 78511a665, 78511a668, 385112668, 38511a668, 72511a668, 74511a668, 345112668, 34511a668, 345119668, 745119668, 725119668, 785119668, 385119668, 385119667, 785119667, 725119667, 745119667, 345119667, 345112667, 385112667, 345112367, 345112368, 34511a368, 345112a67, 74511a368, 745119368, 745119367, 745119a67, 345119a67, 345119368, 345119367, 725119367, 725119a67, 785119a67, 725119368, 385119a67, 385119a47, 725119348, 725119347, 725119a47, 785119a47, 74511a378, 345112a77, 345112377, 345112378, 34511a378, 34511a348, 345112347, 345112348, 345112a47, 345119a47, 345119347, 345119348, 745119348, 745119347, 745119a47, 74511a348, 74511a648, 745119648, 345119648, 345112648, 34511a648, 72511a648, 725119648, 785119648, 78511a648, 38511a648, 385112648, 385119648, 385112678, 38511a678, 78511a678, 72511a678, 74511a678, 345112678, 34511a678, 345112677, 385112677, 385112647, 385119647, 785119647, 725119647, 745119647, 345112647, 345119647, 345119645, 345112645, 745119645, 725119645, 785119645, 385112645, 385119645, 38511a645, 78511a645, 72511a645, 74511a645, 34511a645, 34511a675, 74511a675, 72511a675, 78511a675, 38511a675, 385112675, 345112675, 345112375, 345112345, 345119345, 745119345, 74511a345, 34511a345, 34511a375, 74511a375, 74511aa75, 34511aa75, 34511aa45, 74511aa45, 745119a45, 345112a45, 345119a45, 345112a75, 385112a75, 385112a45, 385119a45, 785119a45, 725119a45, 72511aa45, 78511aa45, 38511aa45, 38511aa75, 72511aa75, 78511aa75, 72511a375, 72511a345, 725119345. 
} 
}

\end{document}